\pdfoutput=1
\RequirePackage{amsmath} 
\documentclass[runningheads]{llncs}

\usepackage[noend]{algorithm2e}
\SetKwInput{KwData}{Input}
\SetKwInput{KwResult}{Output}

\usepackage{tikz}

\usepackage{hyperref}

\usepackage{arydshln}

\usepackage{amsmath,amssymb,amsfonts}
\usepackage{color} 

\newcommand{\ff}[1]{\mathbb{F}_{#1}} 
\newcommand{\Kf}{\mathbb{K}}
\newcommand{\Fq}{\ff{q}}
\newcommand{\Fqm}{\ff{q^m}}
\newcommand{\NN}{\mathbb{N}} 
\newcommand{\any}{*} 
\newcommand{\rank}[1]{\operatorname{Rank}\mathchoice{\left(#1\right)}{(#1)}{(#1)}{(#1)}} 
\newcommand{\rw}[1]{\left| #1 \right|}
\DeclareMathOperator{\Mat}{Mat}
\DeclareMathOperator{\Jac}{Jac}
\DeclareMathOperator{\MaxMinors}{MaxMinors}
\newcommand{\Sign}[2]{\sigma_{#2}{(#1)}}
\newcommand{\matRing}[3]{#1^{#2 \times #3}} 
\newcommand{\mat}[1]{\boldsymbol{#1}} 
\newcommand{\word}[1]{\vec{\boldsymbol{#1}}} 
\newcommand{\ident}{\mat{I}} 
\newcommand{\trsp}[1]{#1^\mathsf{T}} 
\newcommand{\cv}{\mat{c}}
\newcommand{\ev}{\mat{e}}
\newcommand{\vv}{\mat{v}}
\newcommand{\xv}{\mat{x}}
\newcommand{\yv}{\mat{y}}
\newcommand{\sm}{\mat{S}} 
\newcommand{\zerom}{\mat{0}}
\newcommand{\Am}{\mat{A}}
\newcommand{\Bm}{\mat{B}}
\newcommand{\Km}{\mat{K}}
\newcommand{\Cm}{\mat{C}}

\newcommand{\Gm}{\mat{G}}
\newcommand{\Hm}{\mat{H}}
\renewcommand{\Im}{\mat{I}}

\newcommand{\Mm}{\mat{M}}

\newcommand{\Rm}{\mat{R}}

\newcommand{\Xm}{\mat{X}}
\newcommand{\Ym}{\mat{Y}}

\newcommand{\Cc}{{\mathcal C}}
\newcommand{\cF}{{\mathcal F}}
\newcommand{\ctF}{{\cF}^h}
\newcommand{\tf}{f^h}
\newcommand{\eqdef}{:=} 

\newcommand{\OOtsp}[1]{O(#1)} 
\newcommand{\Thtsp}[1]{\Theta(#1)} 
\newcommand{\OO}[1]{\mathcal{O}\left( #1 \right)}
\newcommand{\Th}[1]{\Theta\left( #1 \right)}

\newcommand{\dff}{d_{\mathrm{ff}}}

\newcommand{\cm}{\mat{C}} 
\newcommand{\err}{\word{e}} 
\newcommand{\synd}{\word{s}} 
\newcommand{\OJLP}{Ourivski-Johansson}

\usepackage[capitalize]{cleveref}

 \makeatletter
\newcommand\bibalias[2]{%
  \@namedef{bibali@#1}{#2}%
}

\newtoks\biba@toks
\newcommand\acite[2][]{%
  \biba@toks{\cite#1}%
  \def\biba@comma{}%
  \def\biba@all{}%
  \@for\biba@one:=#2\do{%
    \@ifundefined{bibali@\biba@one}{%
      \edef\biba@all{\biba@all\biba@comma\biba@one}%
    }{%
      \PackageInfo{bibalias}{%
        Replacing citation `\biba@one' with `\@nameuse{bibali@\biba@one}'
      }%
      \edef\biba@all{\biba@all\biba@comma\@nameuse{bibali@\biba@one}}%
    }%
    \def\biba@comma{,}%
  }%
  \edef\biba@tmp{\the\biba@toks{\biba@all}}%
  \biba@tmp
}
\makeatother

\bibalias{Ouroboros-R}{AABBBDGHZ17}
\bibalias{RQC}{AABBBDGZ17}
\bibalias{RQC2}{AABBBDGZCH19}
\bibalias{ROLLO}{ABDGHRTZABBBO19}
\bibalias{LAKE}{ABDGHRTZ17}
\bibalias{LOCKER}{ABDGHRTZ17a}
\bibalias{RankSign}{AGHRZ17}

\newtheorem{heuristic}{Heuristic}

\begin{document}

\title{An Algebraic Attack on Rank Metric Code-Based Cryptosystems}

\author{
  Magali Bardet\inst{1,2}
  \and
  Pierre Briaud\inst{2}
  \and
  Maxime Bros\inst{3}
  \and
  Philippe Gaborit\inst{3}
  \and
  Vincent Neiger\inst{3}
  \and
  Olivier Ruatta\inst{3}
  \and
  Jean-Pierre Tillich\inst{2}
}

\institute{
  LITIS, University of Rouen Normandie, France
  \and
  Inria, 2 rue Simone Iff, 75012 Paris, France
  \and
  Univ.~Limoges, CNRS, XLIM, UMR 7252, F-87000 Limoges, France
}

\institute{
  LITIS, University of Rouen Normandie, France
  \and
  Inria, 2 rue Simone Iff, 75012 Paris, France
  \and
  Univ.~Limoges, CNRS, XLIM, UMR 7252, F-87000 Limoges, France
}

\maketitle

\begin{abstract}
  The Rank metric decoding problem is the main problem considered in
  cryptography based on codes in the rank metric. Very efficient schemes based
  on this problem or quasi-cyclic versions of it have been proposed recently,
  such as those in the submissions ROLLO and RQC currently at the second round
  of the NIST Post-Quantum Cryptography Standardization Process. While
  combinatorial attacks on this problem have been extensively studied and seem
  now well understood, the situation is not as satisfactory for algebraic
  attacks, for which previous work essentially suggested that they were
  ineffective for cryptographic parameters.
  In this paper, starting from Ourivski and Johansson's algebraic modelling of
  the problem into a system of polynomial equations, we show how to augment
  this system with easily computed equations so that the augmented system is
  solved much faster via Gr\"obner bases. This happens because the augmented
  system has solving degree $r$, $r+1$ or $r+2$ depending on the parameters,
  where $r$ is the rank weight, which we show by extending results from Verbel
  \emph{et al.} (PQCrypto 2019) on systems arising from the MinRank problem;
  with target rank $r$, Verbel \emph{et al.} lower the solving degree to $r+2$,
  and even less for some favorable instances that they call
  ``superdetermined''.  We give complexity bounds for this approach as well as
  practical timings of an implementation using \texttt{magma}. This improves
  upon the previously known complexity estimates for both Gr\"obner basis and
  (non-quantum) combinatorial approaches, and for example leads to an attack in
  200 bits on ROLLO-I-256 whose claimed security was 256 bits.
 
\end{abstract}

\section{Introduction}
\label{sec:intro}

\subsubsection*{Rank metric code-based cryptography.}

In the last decade, rank metric code-based cryptography has proved to be a
powerful alternative to more traditional code-based cryptography based on the
Hamming metric. This thread of research started with the GPT cryptosystem
\cite{GPT91} based on Gabidulin codes \cite{G85}, which are rank metric
analogues of Reed-Solomon codes. However, the strong algebraic structure of
those codes was successfully exploited for attacking the original GPT
cryptosystem and its variants with the Overbeck attack \cite{O05} (see for
example \cite{OTN18} for one of the latest related developments). This has to
be traced back to the algebraic structure of Gabidulin codes that makes masking
extremely difficult; one can draw a parallel with the situation in the Hamming
metric where essentially all McEliece cryptosystems based on Reed-Solomon codes
or variants of them have been broken. However, recently a rank metric analogue
of the NTRU cryptosystem from \cite{HPS98} has been designed and studied,
starting with the pioneering paper \cite{GMRZ13}. Roughly speaking, the NTRU
cryptosystem relies on a lattice that has vectors of rather small Euclidean
norm. It is precisely those vectors that allow an efficient
decoding/deciphering process. The decryption of the cryptosystem proposed in
\cite{GMRZ13} relies on LRPC codes that have rather short vectors in the dual
code, but this time for the rank metric. These vectors are used for decoding in
the rank metric. This cryptosystem can also be viewed as the rank metric
analogue of the MDPC cryptosystem \cite{MTSB12} that relies on short vectors in
the dual code for the Hamming metric. 

This new way of building rank metric code-based cryptosystems has led to a
sequence of proposals \acite{GMRZ13,GRSZ14,LAKE,LOCKER}, culminating in
submissions to the NIST post-quantum competition \acite{Ouroboros-R,RQC}, whose
security relies solely on the decoding problem in rank metric codes with a ring
structure similar to the ones encountered right now in lattice-based
cryptography. Interestingly enough, one can also build signature schemes using
the rank metric; even though early attempts which relied on masking the
structure of a code \acite{GRSZ14a,RankSign} have been broken \cite{DT18b}, a
promising recent approach \cite{ABGHZ19} only considers random matrices without
structural masking.

\subsubsection*{Decoding in rank metric.}

In other words, in rank metric code-based cryptography we are now only left
with assessing the difficulty of the decoding problem for the rank metric. The
rank metric over $\Fq^N$, where $\Fq$ is the finite field of cardinality $q$
and $N=m n$ is a composite integer, consists in viewing elements in this
ambient space as $m \times n$ matrices over \(\Fq\) and considering the
distance $d(\Xm,\Ym)$ between two such matrices $\Xm$ and $\Ym$ as
\[
  d(\Xm,\Ym) = \rank{\Ym-\Xm}.
\]
A (linear matrix) code $\Cc$ in $\Fq^{m \times n}$ is simply a $\Fq$-linear
subspace in $\Fq^{m \times n}$, generated by $K$ matrices $\Mm_1,
\dots, \Mm_K$. The decoding problem for the rank metric at distance $r$ is as
follows: given a matrix $\Ym$ in $\Fq^{m \times n}$ at distance $\leq r$ from
$\Cc$, recover an element $\Mm$ in $\Cc$ at distance $\leq r$ from $\Ym$. This
is precisely the MinRank problem given as input $\Ym$ and $\Mm_1,\dots,\Mm_K$:

\begin{problem}[MinRank]\\
  \emph{Input}: an integer $r \in \NN$ and $K+1$ matrices $\Ym,\mat{M}_1,\dots,\mat{M}_K \in \matRing{\ff{q}}{m}{n}$.\\
  \emph{Output}: field elements $x_1,x_2,\dots,x_K \in \ff{q}$ such that
  \begin{equation*}
    \rank{\Ym - \sum_{i=1}^{K}x_i\Mm_i} \leq r.
  \end{equation*}
\end{problem}
As observed in \cite{BFS99}, the MinRank problem is NP-complete and the best
known algorithms solving it have exponential complexity bounds.

\subsubsection*{Matrix codes specified as $\Fqm$-linear codes.}

However, the trend in rank metric code-based cryptography has been to consider
a particular form of linear matrix codes: they are linear codes of length $n$
over an extension $\Fqm$ of degree $m$ of $\Fq$, that is, $\Fqm$-linear
subspaces of $\Fqm^n$. In the rest of this section, we fix a basis
$(\beta_1,\dots,\beta_m)$ of $\Fqm$ as a $\Fq$-vector space. Then such codes
can be interpreted as matrix codes over $\Fq^{m \times n}$ by viewing a vector
$\xv =(x_1,\dots,x_n) \in \Fqm^n$ as a matrix \(\Mat(\xv) = (X_{ij})_{i,j}\) in
\(\matRing{\Fq}{m}{n}\), where $(X_{ij})_{1 \leq i \leq m}$ is the column
vector formed by the coordinates of $x_j$ in the basis
$(\beta_1,\dots,\beta_m)$, that is, \(x_j = X_{1j} \beta_1 + \cdots + X_{mj}
\beta_m\).

Then the ``rank'' metric $d$ on $\Fqm^n$ is the rank metric on the associated
matrix space, namely
\[
  d(\xv,\yv) \eqdef \rw{\yv - \xv},
  \quad\text{where we define }
  \rw{\xv} \eqdef \rank{\Mat(\xv)}.
\]
An $\Fqm$-linear code $\Cc$ of length $n$ and dimension $k$ over $\Fqm$ specifies a matrix code $\Mat(\Cc)
\eqdef \{\Mat(\cv):\;\cv \in \Cc\}$ in $\Fq^{m \times n}$ of dimension $K
\eqdef m k$ over $\Fq$: it is readily verified that a basis of this
$\Fq$-subspace is given by $(\Mat(\beta_i \cv_j))_{1 \leq i \leq m, 1 \leq j
\leq k}$ where $(\cv_1,\dots,\cv_k)$ is a basis of $\Cc$ over $\Fqm$.

There are several reasons for this trend. On the one hand, the families of matrix codes for
which an efficient decoding algorithm is known are families of $\Fqm$-linear codes. On the other hand,
$\Fqm$-linear codes have a much shorter description than general matrix codes. Indeed, a
matrix code in $\Fq^{m \times n}$ of dimension $K=k m$ can be specified by a
basis of it, which uses $K m n \log(q) = k m^2 n \log(q)$ bits, whereas a
matrix code obtained from an $\Fqm$-linear code of dimension $k$ over $\Fqm$
can be specified by a basis $(\cv_1,\dots,\cv_k)$ of it, which uses $k m n
\log(q)$ bits and thus saves a factor $m$. 

Progress in the design of efficient algorithms for decoding \(\Fqm\)-linear codes suggests that their additional structure may not have a significant impact on the difficulty of
solving the decoding problem. For instance, a generic matrix code over
$\Fq^{m\times n}$ of dimension $K=m k$ can be decoded using the information set
decoder of \cite{GRS16}
within a complexity of the order of $q^{k r}$ when the errors have rank at most
$r$ and $m \geq n$, compared to $q^{kr - m}$ for the decoding of a linear code
over $\Fqm^n$ in the same regime, using a similar decoder \cite{AGHT18}.
Moreover, even if the decoding problem is not known to be NP-complete for these
$\Fqm$-linear codes, there is a randomised reduction to an NP-complete problem
\cite{GZ14} (namely to decoding in the Hamming metric).
Hereafter, we will use the following terminology.

\begin{problem}[$(m,n,k,r)$-decoding problem]\\
  \emph{Input}: an $\Fqm$-basis $(\cv_1,\dots,\cv_k)$ of a subspace $\Cc$ of
  $\Fqm^n$, an integer $r \in \NN$, a vector $\yv \in \Fqm^n$ at distance
  at most \(r\) of \(\Cc\) (i.e.~$\rw{\yv-\cv} \leq r$ for some $\cv \in \Cc$). \\
  \emph{Output}: $\cv \in \Cc$ and $\ev \in \Fqm^n$  such that $\yv=\cv+\ev$
  and $\rw{\ev} \leq r$.
\end{problem}
The region of parameters which is of interest for the NIST submissions
corresponds to $m = \Th{n}$, $k=\Th{n}$ and $r = \Th{\sqrt{n}}$. 

\subsubsection*{Gr\"obner basis techniques for decoding in the rank metric.}

The aforementioned algorithm from \cite{AGHT18} for solving the decoding
problem follows a combinatorial approach pioneered in \cite{OJ02}, which is
related to decoding techniques for the Hamming metric. Another approach
consists in viewing the decoding problem as a particular case of MinRank and
using the algebraic techniques designed for this problem; namely these
techniques use a suitable algebraic modelling of a MinRank instance into a
system of multivariate polynomial equations, and then solve this system with
Gr\"obner basis techniques. Several modellings have been considered, such as
the Kipnis-Shamir modelling \cite{KS99} and the minors modelling (described for
example in \cite{FSS10}); the complexity of solving MinRank using these
modellings has been investigated in \cite{FLP08,FSS10}.
The Kipnis-Shamir modelling boils down to a polynomial system which is
\emph{affine bilinear}. This means that each equation has degree at most $2$
and the set of variables can be partitioned into two sets
$\{x_1,\dots,x_s\}\cup\{y_1,\dots,y_t\}$ such that all monomials of degree $2$
involved in the equations are of the form $x_i y_j$; in other words, the
equations are formed by a quadratic part which is bilinear plus an affine part.
Although the complexity of solving this system can be bounded by that of
solving bilinear systems, which is studied in \cite{FSS11}, the complexity
estimates thus obtained are very pessimistic, as observed experimentally in
\cite{CSV17}. A theoretical explanation of why Gr\"obner basis techniques
perform much better on the Kipnis-Shamir modelling than on generic bilinear
systems was later given in \cite{VBCPS19}. It was also demonstrated there that
the Kipnis-Shamir approach is more efficient than the minors approach on
several multivariable encryption or signature schemes relying on the MinRank
problem. However, the speed-up obtained for the Kipnis-Shamir modelling in the
latter reference mostly comes from the ``superdetermined'' case considered
therein. When applied to the $(m,n,k,r)$-decoding problem, this corresponds to
the case where $m=n$ and $km < nr$; this condition is not met in the decoding
problem instances we are interested in.

Another algebraic approach to solve the $(m,n,k,r)$-decoding problem was
suggested in \cite[\S V.]{GRS16}. It is based on a new modelling specific to
$\Fqm$-linear codes which fundamentally relies on the underlying $\Fqm$-linear
structure and on $q$-polynomials. Also, it results in a system of polynomial
equations that are sparse and have large degree. This approach seems to be
efficient only if $rk$ is not much larger than $n$. 

\subsubsection*{Our contribution.}

If one compares the best known complexity estimates, the algebraic techniques
appear to be less efficient than the combinatorial ones, such as \cite{OJ02}, 
\cite{GRS16}, and \cite{AGHT18} for the parameters of the rank metric schemes
proposed to the NIST \acite{ROLLO,RQC2} or of other rank metric code-based
cryptosystems \cite{L17}. In \cite{LP06}, Levy-dit-Vehel and Perret pioneered
the use of Gr\"obner basis techniques to solve the polynomial system arising in
the Ourivski-Johansson algebraic modelling \cite{OJ02}, with promising
practical timings. In this paper, we follow on from this approach and show how
this polynomial system can be augmented with additional equations that are easy
to compute and bring on a substantial speed-up in the Gr\"obner basis
computation for solving the system. This new algebraic algorithm results in the
best practical efficiency and complexity bounds that are currently known for
the decoding problem; in particular, it significantly improves upon the
above-mentioned combinatorial approaches.

There are several reasons why the Ourivski-Johansson algebraic modelling
improves upon the Kipnis-Shamir one. First, it has the same affine bilinear
structure and a similar number of equations, but it involves much fewer
variables. Indeed, for the case of
interest to us where $m$ and $k$ are in $\Thtsp{n}$ and $r$ is in
$\Thtsp{n^{1/2}}$, the Kipnis-Shamir modelling involves $\Thtsp{n^2}$ equations
and variables, while the Ourivski-Johansson one involves $\Thtsp{n^2}$ equations
and $\Thtsp{n^{3/2}}$ variables.
Second, this modelling naturally leads to what corresponds to reducing by one
the value of $r$, as explained in Section \ref{sec:modellings}. Third,
and most importantly, the main properties that ensure that the Kipnis-Shamir
modelling behaves much better with respect to Gr\"obner basis techniques than
generic bilinear systems also hold for the Ourivski-Johansson modelling. In
essence, this is due to a \emph{solving degree} which is remarkably low: at
most $r+2$ for the former modelling and at most $r+1$ for the latter. Recall
that the solving degree indicates the maximum degree reached during a Gr\"obner
basis computation; it is known to be a strong predictor of the complexity of
the most expensive step in a Gr\"obner basis computation and has been widely
used for this purpose with confirmations via numerical experiments, see for
instance \cite{GJS06,DG10,DK11,DS13,DY13,VBCPS19}.

To prove the third point, we start from the result about degree falls at the
core of \cite{VBCPS19}, which is based on work from \cite{FSS11}, and we extend
it to a more general setting which includes the Ourivski-Johansson modelling.
In our case, these degree falls mean that from the initial system of quadratic
equations $f_i=0$ of the Ourivski-Johansson modelling, we are able to build
many new equations of degree $r$ that are combinations $\sum_i f_i g_{ij}=0$
where the $g_{ij}$'s are polynomials of degree $r-1$ involved in the $j$-th new
equation.  We also prove that, when the parameters satisfy the condition
\begin{equation}
  \label{eq:highly-overdetermined}
  m \binom{n-k-1}{r} \ge \binom{n}{r} -1,
\end{equation}
by using that these polynomials $\sum_i f_i g_{ij}$ can be expressed as linear
combinations of only a few other polynomials, we can perform suitable linear
combinations of the equations $\sum_i f_i g_{ij}=0$'s giving
$\binom{n-1}{r-1} - 1$
equations  of degree $r-1$. All these polynomial
combinations are easily computed from the initial quadratic equations. By
adding these equations and then performing Gr\"obner basis computations on the
augmented system, we observe that experimentally the Gr\"obner basis algorithm
behaves as expected from the degree fall heuristic:
\begin{itemize}
  \item if \eqref{eq:highly-overdetermined} holds, this degree is $r$ and the
    overall complexity is
    $\OO{\left(\frac{((m+n)r)^{r}}{r!}\right)^\omega}$ operations in $\ff{q}$.
  \item if \eqref{eq:highly-overdetermined} does not hold, the maximum
    degree reached in the Gr\"obner basis computation is $r+1$ (in
    some intermediate cases), or $r+2$, leading to an overall
    complexity of at most
    $\OO{\left(\frac{((m+n)r)^{r+1}}{(r+1)!}\right)^\omega}$ (resp.
    $\OO{\left(\frac{((m+n)r)^{r+2}}{(r+2)!}\right)^\omega}$)
    operations in \(\ff{q}\), where $\omega$ is the exponent of matrix
    multiplication;
\end{itemize}
Note that for a majority of parameters proposed in \acite{ROLLO,RQC2}, the
condition \eqref{eq:highly-overdetermined} holds. Taking for $\omega$ the
smallest value currently achievable in practice, which is $\omega\approx2.8$
via Strassen's algorithm, this leads to an attack on the cryptosystems proposed in
the aforementioned NIST submissions which is in all cases below the claimed classical
security levels.

\section{Notation}
\label{sec:notation}

In the whole paper, we use the following notation and definitions:
\begin{itemize}
  \item Matrices and vectors are written in boldface font $\mat M$.
  \item For a matrix $\Mm$ its entry in row $i$ and column $j$ is denoted by 
  $\Mm[i,j]$.
  \item The transpose of a matrix \(\mat{M}\) is denoted by \(\trsp{\mat{M}}\).
  \item For a given ring \(\mathcal{R}\), the space of matrices with $m$
    rows and $n$ columns and coefficients in \(\mathcal{R}\) is denoted by
    $\mathcal{R}^{m\times n}$.
  \item For $\Mm\in\mathcal R^{m\times n}$, we denote by $vec_{row}(\Mm)$ the column vector formed by concatenating the rows of $\Mm$, i.e. $vec_{row}(\Mm)=\trsp{
      \begin{pmatrix}
        \Mm_{\{1\},\any} & \dots & \Mm_{\{n\},\any}
      \end{pmatrix}
}$.
\item For $\Mm\in\mathcal R^{m\times n}$, we denote by $vec_{col}(\Mm)$ the column vector formed by concatenating the columns of $\Mm$, i.e. $vec_{col}(\Mm)=vec_{row}(\trsp{        \Mm})$.
  \item $\{1..n\}$ stands for the set of integers from $1$ to $n$, and for any subset $J\subset\{k+1..n\}$, we denote by $J-k$ the set $J-k = \{j-k : j \in J\}\subset\{1..n-k\}$.
  \item For two subsets $I\subset\{1..m\}$ and $J\subset\{1..n\}$,
    we write $\mat M_{I,J}$ for the submatrix of $\mat M$ formed by its rows
    (resp.~columns) with index in $I$ (resp.~$J$).
  \item We use the shorthand notation $\mat M_{\any,J} = \mat{M}_{\{1..m\},J}$
    and $\mat M_{I,\any} = \mat{M}_{I,\{1..n\}}$, where \(\mat M\) has \(m\)
    rows and \(n\) columns.
  \item $\ff{q}$ is a finite field of size $q$, and  $\alpha \in \ff{q^m}$ is a primitive element, so that
    $(1,\alpha,\dots,\alpha^{m-1})$ is a basis of $\ff{q^m}$ as an
    $\ff{q}$-vector space. For $\beta \in \ff{q^m}$, we denote by $[\alpha^{i-1}]\beta$ its $i$th coordinate in this basis. 
  \item For $\vec{v}=(v_1,\ldots,v_n) \in \ff{q^m}^n$. The \emph{support} of \(\vec v\) is
    the $\ff{q}$-vector subspace of $\ff{q^m}$ spanned by the vectors
    $v_1,\ldots,v_n$. Thus this support is the column space of the matrix
    \(\Mat(\vec{v})\) associated to \(\vec{v}\) (for any choice of basis), and
    its dimension is precisely $\rank{\Mat(\vec{v})}$.
  \item An $[n,k]$ $\Fqm$-linear code is an $\Fqm$-linear subspace of $\Fqm^n$
    of dimension $k$ endowed with the rank metric.
\end{itemize}

\section{Algebraic modellings of the decoding problem}
\label{sec:modellings}

In what follows, parameters are chosen in the cryptographically relevant region mentionned in the introduction, say $m = \Th{n}$, $k=\Th{n}$ and $r = \Th{\sqrt{n}}$. Decoding instances will then have a single solution $\err$. For simplicity, we assume that the rank of
$\err$ is exactly $r$; in general one can run the algorithm for increasing
values of the target rank up to $r$, until a solution is found, and the most
expensive step will correspond to the largest considered rank. We consider 
here the $(m,n,k,r)$-decoding problem for the code $\Cc$ and assume we have received 
$\yv \in \Fqm^n$ at distance $r$ from $\Cc$ and look for $\cv \in \Cc$ and $\ev$ such that 
$\yv = \cv + \ev$ and $|\ev|=r$.

\subsection{Solving the MinRank instance using Kipnis-Shamir's modelling}
\label{sec:modelling:ks}

As explained in \cref{sec:intro}, a possible approach to perform the decoding
 is to solve the underlying MinRank
instance with $km+1$ matrices in $\ff{q}^{m \times n}$; this is done by 
introducing $\Mm_0 \eqdef \Mat(\yv)$ and 
$\Mm_1,\dots,\Mm_{km}$ which is an $\Fq$-basis of $\Mat(\Cc)$.
Several methods have been developed, and so far the Kipnis-Shamir modelling \cite{KS99} 
seems to be the
most efficient to solve this MinRank instance. We want to find $(z_0,\dots,z_{km})$ 
in $\Fq^{mk+1}$ such that 
$\rank{\sum_{i=0}^{km} z_i \Mm_i}=r$. $(z_0,z_1,\dots,z_{km})$
is a solution to the MinRank problem if and only if the right kernel of
$\sum_{i=0}^{km} {z_i}\mat{M}_i$ contains a subspace of dimension $n-r$ of
$\ff{q}^{n}$. With high probability, a basis of such a space can be written in
systematic form, that is, in the form $\left(\begin{smallmatrix} \ident_{n-r} \\
\mat{K} \end{smallmatrix}\right)$. Thus we have to solve the system
\begin{equation}
  \left(\sum_{i=0}^{km} {z_i}\mat{M}_i\right)\begin{pmatrix}
  \ident_{n-r} \\ \mat{K}
\end{pmatrix}=0, 
\end{equation}
over \(\ff{q}\), where $\mat{K}$ is an $r \times (n-r)$ matrix of
indeterminates. This system is affine bilinear and has $m(n-r)$ equations and
\(km+1+r(n-r)\) variables, which are $z_0,z_1,\dots,z_{km}$ and the $r(n-r)$
entries of $\mat{K}$; each equation has a bilinear part as well as a linear
part which only involves the variables $z_i$.

\subsection{Syndrome modelling}
\label{sec:modelling:syndrome}

We recall here the modelling considered in \acite{ROLLO,RQC}.
Let $\mat{H}$ be a parity-check matrix of  $\Cc$, i.e.
$$\Cc = \{\cv \in \Fqm^n:\cv \trsp{\Hm}=\zerom\}.$$
The \((m,n,k,r)\)-decoding problem can be algebraically described by the system
$\err \trsp{\mat{H}} = \synd$ where $\err \in \Fqm^n$ has
rank $r$ and $\synd\in \Fqm^{(n-k)}$ is given by $\synd \eqdef \yv \trsp{\Hm}$. Let
$(S_1,\dots,S_r)\in\ff{q^m}^r$ be a basis of the support of $\err$; then, $\err
= (S_1 \;\; \cdots \;\; S_r)\cm$, where $\cm\in\matRing{\ff{q}}{r}{n}$ is the
matrix of the coordinates of $\err$ in the basis $(S_1,\dots,S_r)$. Then
expressing the elements $S_i$ in the basis $(1,\alpha,\dots,\alpha^{m-1})$ of
$\ff{q^m}$ over $\ff{q}$ yields $(S_1 \;\; \cdots \;\; S_r) = (1 \;\; \alpha
\;\; \cdots \;\; \alpha^{m-1}) \sm$ for some matrix $\sm \in
\matRing{\ff{q}}{m}{r}$. Thus, the system is rewritten as
\begin{equation}
  \label{eq:syndromfqm}
  \begin{pmatrix}
    1 & \alpha & \cdots & \alpha^{m-1}
  \end{pmatrix}
  \sm \cm \trsp{\mat{H}}
  = \synd,
  \;\; \text{over } \ff{q^m} \text{ with solutions in } \ff{q}.
\end{equation}
This polynomial system, that we refer to as the \emph{syndrome modelling}, has $m(n-k)$ equations and $mr+nr$ variables 
when it is written over $\ff{q}$. It is affine bilinear
(without terms of degree $1$) with respect to the two sets of variables coming from the
support and from the coordinates of the error. Besides, this system admits
$(q^r-1)(q^r-q) \cdots (q^r-q^{r-1})$ solutions since this is the number of
bases of the support. These solutions to the system all correspond to the same
unique solution $\err$ of the initial decoding problem. We can easily impose a unique solution by fixing some of the unknowns  as in the 
Kipnis-Shamir modelling, or as has been done in the \OJLP\ modelling that we will present next. It is worthwhile to note that this kind 
of modelling has, as the Kipnis-Shamir modelling, $\Th{n^2}$ equations for our choice of parameters but significantly fewer 
variables since we now have only $\Th{n^{3/2}}$ unknowns. The \OJLP's modelling will be a related modelling that gives a further improvement. 
\subsection{Ourivski-Johansson's modelling}
\label{sec:modelling:oj}

We now describe the algebraic modelling considered in the rest of this paper,
which is basically Ourivski and Johansson's one \cite{OJ02}.
It can be viewed as an homogenising trick. Instead of looking 
for $\cv \in \Cc$ and $\ev$ of rank $r$ that satisfy $\yv =\cv+\ev$, or what is the same for $\cv \in \Cc$ such that 
$|\cv+\yv|=r$, we look for $\cv \in \Cc$ and $\lambda \in \Fqm$ such that 
\begin{equation}\label{eq:homogenisingtrick} 
|\cv+ \lambda \yv|=r.
\end{equation} 
It is precisely here that the $\Fqm$-linearity of $\Cc$ is used in a crucial way. 
Once we have found such a $\cv$ and $\lambda$, we have found a $\cv+ \lambda \yv$ such 
that $\cv+\lambda \yv = \mu \ev$ for some non-zero $\mu \in \Fqm$ from which we deduce easily $\ev$.
The point of proceeding this way is that there are $q^m-1$ solutions to \eqref{eq:homogenisingtrick} and that 
this allows us to fix more unknowns in the algebraic system. Another point of view \cite[Sec.\,2]{OJ02} is to say that we introduce the code 
$\tilde{\Cc} \eqdef \Cc + \langle \yv \rangle$ and that we look for  a rank $r$ word in $\tilde{\mathcal C}$, since all such
words are precisely the multiples \(\lambda\err\) for nonzero
\(\lambda\in\Fqm\) of the error \(\err\) we are looking for. Let $\tilde{\Gm} =
\begin{pmatrix}
\ident_{k+1}& \mat{R}
\end{pmatrix}$ (resp.  $\tilde{\Hm} =
\begin{pmatrix}
  -\trsp{\mat R} & \Im_{n-k-1}
\end{pmatrix}
$) be the generator matrix in systematic form (resp. a parity-check
matrix) of the extended code $\tilde{\mathcal C}$; note that for a
vector $\vec{v}$, we have \(\vec{v} \in \tilde{\mathcal{C}}\) if and
only if \(\vec{v} \trsp{\tilde{\Hm}} = 0\). Using the notation
$\err= (1 \;\; \alpha \;\; \cdots \;\; \alpha^{m-1}) \sm \cm$ as
above, and writing $\mat{C} = (\mat{C}_1 \;\; \mat{C}_2)$ with
$\mat{C}_1\in\matRing{\ff{q}}{r}{(k+1)}$ and
$\mat{C}_2\in\matRing{\ff{q}}{r}{(n-k-1)}$, the fact that
\(\err \in \tilde{\mathcal{C}}\) yields the system
\begin{equation}
  \label{eq:oj-modelling}
  \begin{pmatrix}
    1 & \alpha & \cdots & \alpha^{m-1}
  \end{pmatrix}
  \sm \left( \mat{C}_2 - \mat{C}_1\mat{R}\right) = 0,
  \;\; \text{over } \ff{q^m} \text{ with solutions in } \ff{q}.
\end{equation}
Since all multiples \(\lambda \err\) are solutions of this system, we
can specify the first column of $\mat{C}$ to
$\trsp{(1 \; 0 \; \cdots \; 0)}$. In this way, there is a single
$\lambda \err$ satisfying these constraints: the one where $\lambda$
is the inverse of the first coordinate of $\err$ (assuming it is
nonzero, see below). The system still admits several solutions which
correspond to different bases of the support of \(\lambda\err\). To
fix one basis of this support, similarly
to what is done in \cite[Sec.\,3]{OJ02}, we can specify $S_1=1$, or
equivalently, set the first column of $\mat{S}$ to be
$\trsp{(1 \; 0 \; \cdots \; 0)}$, and take an \(r\times r\) invertible
submatrix of $\mat{S}$ and specify it to be the identity matrix; thus
the system has a single solution.

Doing so, the resulting system is affine bilinear (without constant
term), with $(n-k-1)m$ equations and $(m-1)r+nr$ variables, and has a
unique solution.

For the sake of presentation, in \cref{sec:jac} we present our results assuming
that the first coordinate of \(\err\) is nonzero and that the top \(r\times r\)
block of \(\mat{S}\) is invertible; these results are easily extended to the
general case. Under these assumptions, our system can be rewritten as follows:
\begin{equation}
  \label{eqn:oj-modelling-specialized}
\mathcal F  = \left\{
  \begin{pmatrix}
    1 & \alpha & \cdots & \alpha^{m-1}
  \end{pmatrix}
  \left(\begin{array}{c|c}
      \multicolumn{2}{c}{\mat I_r} \\
      \hline
      \mat{0} & \,\;\sm'\;\, \vphantom{\hat\sm} \\
  \end{array}\right)
  \left(\cm_2 - 
  \left(\begin{array}{c|}
     1\, \\ \mat{0}\,
    \end{array} \;\, \cm_1'\right)
  \mat{R}\right)
\right\},
\end{equation}
where $\sm'$ is the $(m-r)\times (r-1)$ submatrix $\sm_{\{r+1..m\},\{2..r\}}$
and $\cm_1'$ is the $r\times k$ submatrix $\cm_{\any,\{2..k+1\}}$. 
We call the entries of $\sm'$ the {\em support variables} whereas the entries of $\cm_1'$ and $\cm_2$ are called  the 
{\em coefficient variables}. 
In
\cref{sec:results:complexity} we give a procedure to handle the general case, by
making several attempts to find the invertible block of \(\mat{S}\) and
a nonzero component of \(\err\).

\section{Gr\"obner bases and degree falls}
\label{sec:syzygies}

We refer to~\cite{CLO01} for basic definitions and properties of
monomial orderings and Gr\"obner bases.

\paragraph{\bf Field equations and monomial ordering}
Since we are looking for solutions in $\mathbb{F}_q$, we augment the
polynomial system we want to solve with the field equations,
that is, the equation \(x_i^q-x_i=0\) for each variable $x_i$ arising in the
system. In our case, as the system we consider in practice has mainly only one
solution in $\mathbb{F}_q$ (see Section~\ref{sec:results}), the ideal of the system with the field
equations is radical, and for any monomial ordering the reduced
Gr\"obner basis is the set of linear polynomials $\{x_i-a_i\}_i$,
where $\{x_i\}_i$ are the variables and $a_i\in\mathbb{F}_q$ is the
$i$-th coordinate of the solution. The classical approach consists in
computing the Gr\"obner basis with respect to a degree-reverse
lexicographic order (grevlex), that will keep the degree of the
polynomials as small as possible during the computation, and behaves
usually better than other monomial orderings in terms of complexity.

\paragraph{\bf Generic Gr\"obner bases algorithms and their link with linear algebra}
Since the first descriptions of algorithms to compute Gr\"obner bases
\cite{BUC65}, far more efficient algorithms have been developed. On
the one hand, substantial practical speed-ups were achieved by
incorporating and accelerating fast linear algebra operations such as
Gaussian elimination on the Macaulay matrices, which are sparse and
structured (see Faug\`ere's F4 algorithm~\cite{F99}, variants of the
XL algorithm~\cite{CKPS00}, and for instance GBLA~\cite{BEFLM16}). We
recall that the Macaulay matrix in degree $d$ of a homogeneous system
$(f_i)_i$ is the matrix whose columns correspond to the
monomials of degree $d$ sorted in descending order w.r.t.~a chosen
monomial ordering, whose rows correspond to the polynomials $t f_i$ for
all $i$ where $t$ is a monomial of degree $d-\deg(f_i)$, and whose
entry in row $tf_i$ and column $u$ is the coefficient of the monomial
$u$ in the polynomial $tf_i$. In the case of a system containing field
equations, we consider compact Macaulay matrices, where all monomials
are reduced w.r.t.~the field equations. For an affine system, the
Macaulay matrix in degree $d$ contains all polynomials $\{t f_i\}$ for
$\deg(t f_i)\le d$ and the columns are the monomials of degree less
than or equal to $d$.

The approaches from F4 or XL are similar in that they both compute row
echelon forms of some submatrices of Macaulay matrices for some
given degree; in fact, it was proven in \cite{AFIKS04} that the XL
algorithm computes a so-called $d$-Gr\"obner basis, which is a basis
of the initial system where all computations in degree larger than $d$
are ignored, and that one can rephrase the original XL algorithm in
terms of the original F4 algorithm.

Now, many variants of these algorithms have been designed to tackle
specific families of polynomial systems, and it seems that none of
them performs always better than the others. In our experimental
considerations, we rely on the implementation of the F4 algorithm
which is available in \texttt{magma V2.22-2} and is recognised for its
efficiency.

On the other hand, improvements have been obtained by refining criteria which
allow one to avoid useless computations (avoiding to consider monomials that
cannot appear, a priori detection of reductions to zero as in the F5
algorithm~\cite{F02} and other signature-based algorithms that followed,
see~\cite{EF17} for a survey).

\paragraph{\bf Complexity analysis for homogeneous systems}
For \emph{homogeneous} systems, and for a graded monomial ordering,
the complexity of these algorithms in terms of arithmetic operations
is dominated by the cost of the row echelon forms on all Macaulay
matrices up to degree $d$, where $d$ is the  largest  degree of a
polynomial in the reduced Gr\"obner basis\footnote{If the system
  contains redundant polynomials of degree larger than $d$, additional
  operations are needed to check that those polynomials reduce to zero
  w.r.t. the Gr\"obner basis, but this has usually a negligible cost.}.  This
degree $d$ is called the \emph{index of regularity}, or \emph{degree
  of regularity}, and it only depends on the ideal generated by the
system, not on the specific generators forming the system.
Some algorithms may need to go beyond degree $d$ to check that no new
polynomials will be produced, like the XL Algorithm or the F4 Algorithm without
the F5 criteria, but those computations may be avoided if one knows in advance
the degree of regularity of the system.
This parameter can be precisely estimated for different families of
generic systems, using the notions of regularity, of semi-regularity
in the over-determined case, and of bi-regularity in the bilinear
case~\cite{B04,BFSY05,BFS15,FSS11}. However, those bounds may be very
pessimistic for other specific (sub-)families of systems, and
deriving estimations in this situation is difficult a priori, in
particular for affine systems.

\begin{definition}
  Let $(f_i)_i$ be (non necessarily homogeneous) polynomials in a
  polynomial ring \(\mathcal R\). A syzygy is a vector $(s_i)_i$,
  $s_i\in \mathcal R$ such that $\sum_i s_if_i=0$. The degree of the
  syzygy is defined as $\max_i(\deg(f_i)+\deg(s_i))$. The set of all
  syzygies of $(f_i)_i$ is an $\mathcal R$-module called the syzygy
  module of $(f_i)_i$.
\end{definition}

For a given family of systems, there are syzygies that occur for any system in
the family. For instance, for any system $(f_i)_i$, the syzygy module contains
the \(\mathcal{R}\)-module spanned by the so-called \emph{trivial syzygies}
$(e_jf_i - e_i f_j)_{i,j}$, where $e_i$ is the coordinate vector with \(1\) at
index $i$. A system is called \emph{regular} if its syzygy module is generated
by these trivial syzygies.

Let us consider the particular case of a zero-dimensional system
$(f_i)_i$ of homogeneous polynomials, generating an ideal $I$. As the
system is homogenous and has a finite number of solution, then it must
have only $0$ as a solution (with maybe some multiplicities). In this
case, the degree of regularity of the system is the lowest integer
$d_{\mathrm{reg}}$ such that all monomials of degree
$d_{\mathrm{reg}}$ are in the ideal of leading terms of $I$
(see~\cite{B04,BFSY05}).
Such a system is called \emph{semi-regular}
if the set of its syzygies of degree less than $d_{\mathrm{reg}}(I)$ is exactly
the set of trivial syzygies of degree less than $d_{\mathrm{reg}}(I)$. Note
that there may be non-trivial syzygies in degree $d_{\mathrm{reg}}(I)$,
which may be different for each system.  As a consequence, all polynomials
occurring in the computation of a Gr\"obner basis have degree $\le d_{\mathrm{reg}}$
and the arithmetic complexity is bounded by the cost of the row
echelon form on the Macaulay matrices in degree $\le d_{\mathrm{reg}}$.

\paragraph{\bf Complexity analysis for affine systems}
For affine systems, things are different. The degree of regularity can
be defined in the same way w.r.t. the Gr\"obner basis for a grevlex ordering.
But is not any more related to the complexity of the computation: for
instance, a system with only one solution will have a degree of
regularity equal to 1. 
We need another parameter to control
the complexity of the computation.

Let $(f_i)_i$ be a system of affine polynomials, and $f_i^h$ the
homogeneous part of highest degree of $f_i$. Let
$I=\langle \{f_i\}_i\rangle$ and
$I^h=\langle \{f_i^h\}_i\rangle$, and let $d^h_{\mathrm{reg}}$ be
the degree of regularity of $I^h$.  What may happen is that,
during the computation of the basis in some degree $d$, some polynomials
of degree less than $d$ may be added to the basis. This
will happen any time a syzygy $(s^h_i)_i$ for $(f^h_i)_i$
of degree $d$ is such that there exists no syzygy $(s_i)_i$ for
$(f_i)_i$ where $s^h_i$ is the homogeneous part of highest
degree of $s_i$. In that case, $\sum_i s^h_if_i$ is a polynomial
of degree less than $d$ (the homogeneous part of highest degree
cancels), that will not be reduced to zero during the Gr\"obner basis
computation since this would give a syzygy $(s_i)_i$ for $(f_i)_i$ with
homogeneous part $(s^h_i)_i$. This phenomenon is called a
\emph{degree fall} in degree $d$, and we will call such syzygies
$(s^h_i)$ that cannot be extended to syzygies for $(f_i)_i$ in the same
degree \emph{partial syzygies}; the corresponding polynomial $\sum_i
s^h_if_i$ is called the {\em residue}.

In cryptographic applications, the \emph{first degree fall} $\dff$ has
been widely used as a parameter controlling the complexity in algebraic
cryptanalysis, for instance in the study of some HFE-type systems
\cite{DG10,GJS06,DH11} and Kipnis-Shamir systems \cite{VBCPS19}. This
first degree fall is simply the smallest $d$ such that
there exists a degree fall in degree $d$ on $(f_i)_i$, and this
quantity does depend on $(f_i)_i$: it might be different for
another set of generators of the same ideal.  Still, this notion takes
on its full meaning while computing a Gr\"obner basis for a graded
ordering, if we admit that the algorithm terminates shortly after
reaching the first degree fall and without considering polynomials of
higher degree. This can happen for some families of systems, as explained in
the next paragraph, but there are examples of systems where the first
degree fall $\dff$ is not the maximal degree reached during the computation,
in which case it is not related to the complexity of the computation.

If the system $(f_i^h)_i$ is semi-regular, then the computation
in degree less than $d^h_{\mathrm{reg}}$ will act as if the polynomials
where homogeneous: there cannot be degree falls, as they would
correspond to syzygies for the system $(f_i^h)_i$ that is assumed
to be semi-regular. In degree $d^h_{\mathrm{reg}}$, degree falls will
occur for the first time, but at this point the remainder of the
computation is negligible compared to the previous ones: by definition
of $d^h_{\mathrm{reg}}$, all monomials of degree $d^h_{\mathrm{reg}}$ are
leading terms of polynomials in the basis, and the remaining steps in
the computation will necessarily deal with polynomials of degree at
most $d^h_{\mathrm{reg}}$. 
Hence, the computations are almost the same as the ones for
$(f_i^h)_i$, and the complexity is controlled by
$d^h_{\mathrm{reg}}$, which is here the \emph{first degree fall} for the
system $(f_i)_i$.

The behavior of the computation may be very different if degree falls occur in
a much smaller degree. A good example of what may
happen for particular families of systems is the affine bilinear
case. It is proven in~\cite[Prop.\,5]{FSS11} that a generic affine
bilinear system of $m$ equations
$(f_1,\dots,f_m)\in\mathbb K[x_1,\dots,x_{n_x},y_1,\dots,y_{n_y}]$
in $n_x+n_y\ge m$ variables is regular. In particular, the Macaulay bound
$d_{\mathrm{reg}} \le n_x+n_y+1$ applies~\cite{L83}. However, it was also proven
in~\cite[Thm.\,6]{FSS11} that for a zero-dimensional affine bilinear
system ($m=n_x+n_y$), $d_{\mathrm{reg}}$ satisfies a much sharper inequality
$d_{\mathrm{reg}}\le \min(n_x+1,n_y+1)$. The reason is that (homogeneous)
bilinear systems are not regular, but the syzygy module of those
systems is well understood~\cite{FSS11}. In particular, there are
syzygies for $(f_i^h)_i$ coming from Jacobian matrices, that are
partial syzygies for $(f_i)_i$ and produce degree falls. 

For affine systems, that are mainly encountered in cryptographic
applications, and in particular for systems coming from a product of
matrices whose coefficients are the variables of the system, the
Jacobian matrices have a very particular shape that is easily described, and
leads to a series of degree falls that reduces the degree of regularity of
those systems. This is explained in detail in Section~\ref{sec:jac}.

\section{Degree falls and low degree equations}
\label{sec:jac}

\subsection{Degree falls from the kernel of the Jacobian}

\paragraph{\bf Fundamental results from \cite{FSS11,VBCPS19}.}
It has been realized in \cite{VBCPS19} that the first degree fall in the Kipnis and Shamir modelling can be traced back to 
partial syzygies obtained from low
degree vectors in the kernel of the Jacobian of the bilinear part of a system either with respect to the kernel variables or the linear variables.
This argument can also be adapted to our case and Jacobians with respect to the support variables are relevant here.
To understand the relevance of the Jacobians for bilinear affine systems over some field $\Kf$ in general, 
consider a bilinear affine system $\mathcal F = \{f_1,\dots,f_M\} \subset \mathbb K[s_1,\dots,s_{t_s},c_1,\dots,c_{t_c}]$ of $M$ equations in $t_s$ variables $s$ and $t_c$ variables $c$. We denote by $\ctF \eqdef \{\tf_1,\dots,\tf_M\}$ the bilinear part of these equations. 
In other words each $f_i$ can be written as
\[
f_i = \tf_i + r_i,
\]
where each $r_i$ is affine and ${f_i}^h$ is bilinear with respect to $\{s_1,\dots,s_{t_s}\}\cup\{c_1,\dots,c_{t_c}\}$.
We define the Jacobian matrices associated to $\ctF$ as
\[
\begin{array}{lcr}
\Jac_{\sm}(\ctF) &=& \begin{pmatrix}
\frac{\partial \tf_1}{\partial s_1} & \hdots & \frac{\partial \tf_1}{\partial s_{t_s}}\\
\vdots & \vdots & \vdots \\
\frac{\partial \tf_M}{\partial s_1} & \hdots & \frac{\partial \tf_M}{\partial s_{t_s}}
\end{pmatrix}
\end{array}
\;\;\;\text{and}\;\;\;
\begin{array}{lcr}
\Jac_{\mat C}(\ctF) &=& \begin{pmatrix}
\frac{\partial \tf_1}{\partial c_1} & \hdots & \frac{\partial \tf_1}{\partial c_{t_c}}\\
\vdots & \vdots & \vdots \\
\frac{\partial \tf_M}{\partial c_1} & \hdots & \frac{\partial \tf_M}{\partial c_{t_c}}
\end{pmatrix}
\end{array}
.
\]

Note that $\Jac_{\sm}(\ctF)$ is a matrix with linear entries in $\Kf[c_1,\dots,c_{t_c}]$ whereas 
$\Jac_{\mat C}(\ctF)$ is a matrix with linear entries in $\Kf[s_1,\dots,s_{t_s}]$.
As shown in \cite{VBCPS19}[Prop. 1 \& 2] vectors in the left kernel of these Jacobians yield
partial syzygies. This is essentially a consequence of the following identities that are easily verified:
\[
  \Jac_{\sm}(\ctF) \begin{pmatrix} s_1 \\ \vdots \\ s_{t_s}\end{pmatrix}
  = \begin{pmatrix} \tf_1 \\ \vdots \\ \tf_M\end{pmatrix}
  \;\;\;\text{and}\;\;\;
  \Jac_{\mat C}(\ctF) \begin{pmatrix} c_1 \\ \vdots \\ c_{t_c}\end{pmatrix}
  = \begin{pmatrix} \tf_1 \\ \vdots \\ \tf_M\end{pmatrix}.
\]
For instance, a vector $(g_1, \ldots, g_M)$ in the left kernel of
$\Jac_{\cm}(\ctF)$ is a syzygy for $\ctF$, as it satisfies
\[
  \sum_{i=1}^M g_i \tf_i
  = ( g_1 \,\cdots\, g_M ) \begin{pmatrix} \tf_1 \\ \vdots \\ \tf_M\end{pmatrix}
  = (g_1 \,\cdots\, g_M) \Jac_{\mat C}(\ctF) \begin{pmatrix} c_1 \\ \vdots \\ c_{t_c} \end{pmatrix} = 0.
\]
This gives typically a degree fall for $\cF$ at degree $2+\max(\deg g_i)$, with the corresponding residue given by 
\[
  \sum_{i=1}^M g_i f_i = \sum_{i=1}^M g_i \tf_i + \sum_{i=1}^M g_i r_i = \sum_{i=1}^M g_i r_i.
\]
These Jacobians are matrices with entries that are linear forms. The kernel of
such matrices is well understood as shown by the next result.

\begin{theorem}[\cite{FSS11}]\label{thm:vj}
Let $\Mm$ be an $M \times t$ matrix of linear forms in $\Kf[s_1,\dots,s_{t_s}]$. 
If $t < M $,
then generically the left kernel of $\Mm$
is generated by vectors whose coefficients are maximal minors of $\Mm$, specifically
vectors of the form
\begin{equation*}
    \vec{V}_{J} = (\dots,
      \underbrace{0}_{j \notin J},\dots,
      \underbrace{
        (-1)^{l+1} 
      \det(\Mm_{J\backslash \{j\},{\any}})}_{j\in J, j = j_l},\dots
      )_{1\le j \le M}
    \end{equation*}
    where $J = \{j_1<j_2<\dots<j_{t+1}\} \subset \{1,\dots,M\}, \#J = t+1$.
\end{theorem}

A direct use of this result however yields degree falls that occur for very
large degrees, namely at degrees $t_s+2$ or $t_c+2$.  In the case of the
Kipnis-Shamir modelling, the syndrome modelling or the \OJLP\ modelling, due to
the particular form of the systems, degree falls occur at much smaller degrees
than for generic bilinear affine systems. Roughly speaking, the reason is that
the Jacobian of a system coming from a matrix product splits as a tensor
product, as we now explain. This has been realized in \cite{VBCPS19} for the
Kipnis-Shamir modelling, and here we slightly generalize this result in order
to use it for more general modellings, and in particular for the \OJLP\
modelling.

\paragraph{\bf Jacobian matrices of systems coming from matrix products.}
Consider a system $\mat A\mat X \mat Y=\zerom$ where
$\mat A = (a_{i,s})_{1\le i \le m, 1\le s \le p}$,
$\mat X = (x_{s,t})_{1\le s \le p, 1\le t \le r}$ and
$\mat Y = (y_{t,j})_{1\le t \le r, 1\le j \le n}$. The variables
considered for this Jacobian matrix are the $x_{s,t}$. The matrices
$\mat A$ and $\mat Y$ may have polynomial coefficients, but they do
not involve the $x_{s,t}$ variables.  Below, we use the Kronecker
product of two matrices, for example $\mat A\otimes \trsp{\mat Y} =
\begin{pmatrix} a_{i,s}\trsp{\Ym} \end{pmatrix}_{1\le i \le m, 1\le s \le p}$. We use the notations $vec_{row}(\mat A)=\trsp{
      \begin{pmatrix}
        \mat A_{\{1\},\any} & \dots & \mat A_{\{n\},\any}
      \end{pmatrix}
}$ and $vec_{col}(\mat A) = vec_{row}(\trsp{\mat A})$.
\begin{lemma}\label{lemma:jacobian}
  The Jacobian matrix of the system $\mat A \mat X\mat Y=\zerom_{m\times n}$ with
  respect to the variables $\mat X$ can be written, depending on the order of the equations and variables:
  \begin{eqnarray*}
    \Jac_{vec_{col}(\mat X)}(vec_{col}(\mat A\mat X\mat Y)) &=& \trsp{\mat Y}\otimes \mat A \in \mathbb K[\mat A, \mat Y]^{nm\times rp}\\
    \Jac_{vec_{row}(\mat X)}(vec_{row}(\mat A\mat X\mat Y)) &=&  \mat A\otimes \trsp{\mat  Y}\in\mathbb K[\mat A, \mat Y]^{nm\times rp}.
  \end{eqnarray*}
\end{lemma}
\begin{proof}
  For $1\le i \le m$, $1\le j \le n$, the equation in row $i$ and
  column $j$ of $\mat A \mat X\mat Y$ is
  \begin{equation*}
    f_{i,j} = \sum_{s=1}^p \sum_{t=1}^r a_{i,s}  x_{s,t}y_{t,j}.
  \end{equation*}
  We then have, for $1\le s \le p$ and $1\le t \le r$,
  \(\frac{\partial f_{i,j}}{\partial x_{s,t}} =
  a_{i,s} y_{t,j}\)
  so that in row order,
  \begin{eqnarray*}
    \Jac_{x_{s,1},\dots,x_{s,r}}(\{f_{i,1},\dots,f_{i,n}\}) &=& \left(     \frac{\partial f_{i,j}}{\partial x_{s,t}} \right)_{\substack{1\le j \le n\\ 1\le t \le r}} = a_{i,s}\left( y_{t,j} \right)_{\substack{1\le j \le n\\ 1\le t \le r}} = a_{i,s} \trsp{\Ym}.
  \end{eqnarray*}
  The result follows from the definition of the Kronecker product of
  matrices. The proof when the equations and variables are in column
  order is similar.
  \qed
\end{proof}

\paragraph{\bf Application to the Kipnis-Shamir modelling.}
Recall the system:
\begin{equation}
  \left(\sum_{i=1}^{km} {x_i}\mat{M}_i\right)\begin{pmatrix}
  \ident_{n-r} \\ \mat{K}
\end{pmatrix}=\zerom_{m,n-r},
\end{equation}
where $\Mm_i \in \Fq^{m \times n}$ and $\Km$ is an $r \times (n-r)$ matrix of
indeterminates. If we write each $\Mm_i = (\Mm'_i \;\; \Mm''_i)$ with
$\Mm'_i\in \Fq^{m\times (n-r)}$ and $\Mm''_i\in \Fq^{m\times r}$, then we have
\begin{equation}
\sum_{i=1}^{km}  x_i\left(\Mm'_i + \Mm''_i\Km\right) =\zerom_{m,n-r} \tag{KS}
\end{equation}
The bilinear and linear parts of the system are respectively 
$\sum_{i=1}^{km}  x_i \Mm''_i\Km$ and $\sum_{i=1}^{km} x_i\Mm'_i$.
Using Lemma \ref{lemma:jacobian} (with equations in column order), when we
compute the Jacobian with respect to the entries of $\Km$ (the so-called kernel
variables in \cite{VBCPS19}), we obtain 
\begin{eqnarray*}
  \Jac_{vec_{col}(\Km)}(vec_{col}(\sum_{i=1}^{km}  x_i \Mm''_i\Km)) &=& \sum_{i=1}^{km} {x_i}(\Im_{n-r}\otimes {\Mm''_i}) = \Im_{n-r}\otimes \left(\sum_{i=1}^{km} x_i {\Mm''_i}\right).
\end{eqnarray*}
The kernel of $ \Jac_{vec_{col}(\Km)} $ is generated by the vectors
$(\vv_1,\dots,\vv_{n-r})$ with $\vv_l$ in the left kernel of
$\Mm= \sum_{i=1}^{km} x_i{\Mm''_i}$, that should be generated by
$\binom{m}{r+1}$ vectors of minors, according to~Theorem~\ref{thm:vj}.
Hence the kernel of $ \Jac_{vec_{col}(\Km)} $ is generated by
$\binom{m}{r+1}(n-r)$ vectors. It is here that we see the point of
having this tensor product form.  These kernel vectors have entries
that are polynomials of degree $r$ by using Theorem \ref{thm:vj}. This
gives degree falls at degree $r+2$ and yields partial syzygies that
have degree $r+1$. These considerations are a slightly different way
of understanding the results given in \cite[\S 3]{VBCPS19}.  The
syndrome modelling displays a similar behavior, i.e. a degree fall at
$r+2$ for the very same reason as can be readily verified.  Let us
apply now Lemma \ref{lemma:jacobian} to the \OJLP\ modelling.

\paragraph{\bf Application to the \OJLP\ modelling.}
The system here is
\begin{equation}
\mathcal F  = \left\{
  \begin{pmatrix}
    1 & \alpha & \cdots & \alpha^{m-1}
  \end{pmatrix}
  \left(\begin{array}{c|c}
      \multicolumn{2}{c}{\mat I_r} \\
      \hline
      \mat{0} & \,\;\sm'\;\, \vphantom{\hat\sm} \\
  \end{array}\right)
  \left(\cm_2 - 
  \left(\begin{array}{c|}
     1\, \\ \mat{0}\,
    \end{array} \;\, \cm_1'\right)
  \mat{R}\right)
\right\},
\end{equation}
where $\sm'$ is the $(m-r)\times (r-1)$ matrix
$\sm_{\{r+1..m\},\{2..r\}}$ and $\cm_1'$ is the $r\times k$ matrix
$\cm_{\any,\{2..k+1\}}$. We add to $\mathcal F$ the field equations
$\mathcal F_q = \{ s_{i,j}^q-s_{i,j}, r+1\le i \le m, 2\le j \le r,
c_{i,j}^q-c_{i,j}, 1\le i \le r, 2\le j \le n\}$.

With high probability, this system has a unique solution. As we used
the field equations, the ideal $\langle \mathcal F, \mathcal F_q \rangle$ is
radical. The system has $n_{\sm} = (m-r)(r-1)$ variables $\sm$,
$n_{\cm} = (n-1)r$ variables $\cm$, and $n-k-1$ equations
over $\ff{q^m}$, hence  $n_{eq}=(n-k-1)m$ equations over $\ff{q}$, plus the field equations.

Consider the system $\ctF$ formed by the bilinear parts of the equations in $\cF$.
A simple computation shows that
\begin{equation*}
  \ctF = \left\{ \alpha^r
    \begin{pmatrix}
      1 & \alpha & \cdots & \alpha^{m-r-1}
    \end{pmatrix}
    \sm' (\cm_2''-\cm_1''\mat R')
  \right\},
\end{equation*}
where $\cm_2'' = \cm_{\{2..r\},\{k+2..n\}}$, $\cm_1''=
\cm_{\{2..r\},\{2..k+1\}}$ and $\mat R' =  \mat{R}_{\{2..k+1\},\any}$.

If we take the equations and variables in row order,
and use Lemma~\ref{lemma:jacobian}, then
\begin{equation}
  \label{eq:jacS}
  \Jac_{vec_{row}(\sm)}(vec_{row}(\ctF)) = \alpha^r
\begin{pmatrix}
  1 & \alpha & \cdots & \alpha^{m-r-1}
\end{pmatrix}\otimes \trsp{\left(\cm_2''-\cm_1''\mat R'\right)}
\end{equation}
The elements in the left kernel of
$ Jac_{vec_{row}(\sm)}(vec_{row}(\mathcal F^h))$ are those in the right kernel of
$\cm_2''-\cm_1''\mat R'$, and applying Theorem~\ref{thm:vj}, they
belong to the vector space generated by the vectors
${\vec{V}_J}$ for any $J=\{j_1<j_2<\dots<j_{r}\}\subset\{1,\dots,n-k-1\}$ of size $r$
defined by
\begin{equation*}
  \vec V_J = (\dots,
    \underbrace{0}_{j \notin J},\dots,
    \underbrace{
      (-1)^{l+1}
    \det( {\cm_2''-\cm_1''\mat R'}_{\any,J\backslash \{j\}})}_{j=j_l\in J},\dots
    )_{1\le j \le n-k-1}.
\end{equation*}
Each $V_J$  gives a syzygy for $\ctF$ and when applying it to $\cF$ it yields a degree fall in degree $r+1$ because 
the entries of $\vec V_J$ are homogeneous polynomials of degree $r-1$.
The inner product of $\vec V_J$ with the vector of the equations gives
an equation of degree $\leq r$ since the homogeneous part of highest degree cancels,
as has been observed at the beginning of this section.
Now the affine part of the equations $\cF$ is 
$\begin{pmatrix}
  1 & \alpha & \cdots & \alpha^{r-1}
\end{pmatrix} \left(\cm_2-\cm_1\mat R\right) $.  

Writting $\tilde{\Hm} = 
\begin{pmatrix}
  -\trsp{\Rm} & \Im_{n-k-1}
\end{pmatrix}$, then
\[ \det( {\cm_2''-\cm_1''\mat R'}_{\any,J\backslash \{j\}}) = \det((\Cm\trsp{\tilde{\Hm}})_{\{2..r\},J\backslash\{j\}}).\] Using
the reverse of  Laplace's formula expressing a determinant in terms
of minors, we can compute the inner product of the vector
$\vec V_J$ with the $i$th row of $\Cm_2-\Cm_1\Rm = \Cm\trsp{\tilde{\Hm}}$,
that is $0$ for $2\le i$ and $\det((\Cm\trsp{\tilde{\Hm}})_{\any,J})$ for $i=1$.
  
  The product gives   \begin{eqnarray}
    \label{eq:minors}
    \vec V_J  \left(\begin{pmatrix}
        1 & \alpha & \cdots & \alpha^{r-1}
  \end{pmatrix} \trsp{\left(\cm_2-\cm_1\mat R\right)\right)} &=& \vec V_J  \trsp{\left(\cm_2-\cm_1\mat R\right)} 
  \trsp{\begin{pmatrix}
    1 & \alpha & \cdots & \alpha^{r-1}\notag
  \end{pmatrix}}\\
&=& \det (\cm_2 - \cm_1\mat R)_{\any, J}.
  \end{eqnarray}
 This yields a corresponding equation
  that 
  will be reduced to zero by a degree-$(r+1)$ Gr\"obner basis of
  $\mathcal F$.
 Hence the partial syzygies of degree $r$ coming from the degree fall in the
  $(r+1)$-Macaulay matrix are exactly the maximal minors of
  $\cm_2 - \cm_1\mat R$.  We have thus proven that

\begin{theorem}\label{thm:maxminors}
  The equations $\MaxMinors(\cm_2-\cm_1\mat R)=0$, that are the maximal
  minors of the matrix $\cm_2-\cm_1\mat R$, belong to the ideal
  $\langle\mathcal F,\mathcal F_q\rangle$. Moreover, they are reduced
  to zero by a degree $(r+1)$-Gr\"obner basis of
  $\{\mathcal F,\mathcal F_q\}$.
\end{theorem}

\begin{remark}
 If we are only interested in the first part of the theorem about the maximal minors, then there is a simple and direct proof 
which is another illustration of the role of the matrix form of the system. Indeed, let 
$(\sm^*, \cm^*)$ be a solution of $\{\mathcal F,\mathcal F_q\}$,
  then the nonzero vector $
  \begin{pmatrix}
    1 & S_2^* & \cdots S_m^*
  \end{pmatrix} =   \begin{pmatrix}
    1 & \alpha & \cdots & \alpha^{m-1}
  \end{pmatrix}
  \sm^*$ belongs to the left kernel of the matrix
  $\cm_2^* - \cm_1^*\mat R$. Hence this matrix has rank less than $r$,
  and the equations $\MaxMinors(\cm_2-\cm_1\mat R)=0$ are satisfied for
  any solution of the system $\{\mathcal F,\mathcal F_q\}$, which means that the
  equations belong to the ideal 
  $\langle \mathcal F,\mathcal F_q\rangle$ as this ideal is radical.
\end{remark}

  \subsection{Analysis of the ideal $\MaxMinors(\cm_2-\cm_1\mat R)$}

The previous theorem allows us to obtain directly degree $r$ equations without
having to compute first the Macaulay matrix of degree $r+1$. This is a
significant saving when performing the Gr\"obner basis computation. A nice
feature of these equations is that they only involve one part of the unknowns,
namely the coefficient variables.

Moreover all these equations can be expressed by using a limited number of
polynomials as we now show.  Some of them will be of degree $r$, some of them
will be of degree $r-1$.  If we perform Gaussian elimination on these equations
by treating these polynomials as variables and trying to eliminate the ones
corresponding to the polynomials of degree $r$ first, then if the number of
equations we had was greater than the number of polynomials of degree $r$, we
expect to find equations of degree $r-1$. Roughly speaking, when this
phenomenon happens we just have to add all the equations of degree $r-1$ we
obtain in this way to the \OJLP\ modelling and the Gr\"obner basis calculation
will not go beyond degree $r$. 

Let us analyse precisely the behavior we just sketched.
The shape of the equations $\MaxMinors(\cm_2-\cm_1\mat R)=0$ is given
by the following proposition, where by convention
$\det(\Mm_{\emptyset, \emptyset}) = 1$ and the columns of $\Rm$ are indexed by $\{k+2..n\}$:
\begin{proposition}\label{prop:maxminorsshape}
  $\MaxMinors(\cm_2-\cm_1\mat R)$ is a set of $\binom{n-k-1}r$
  polynomials $P_J$, indexed by $J\subset\{k+2..n\}$ of size $r$:
  \begin{eqnarray*}
P_J &=& \sum_{\substack{T_1\subset\{1..k+1\}, T_2\subset J\\\text{such that } T= T_1 \cup T_2 \text{ has size } \#T=r}} (-1)^{\Sign{T_2}{J}}\det(\mat R_{T_1, J\backslash T_2}) \det(\Cm_{\any,T}).\\
  \end{eqnarray*}
  where $\Sign{T_2}{J}$ is an integer depending on $T_2$ and $J$.
  
  If $1\notin T$, the polynomial $\det(\cm_{\any,T})$ is homogeneous of degree $r$ and contains $r!$ monomials; if $1\in T$, $\det(\cm_{\any,T})$ is homogeneous of degree $r-1$ and contains $(r-1)!$ monomials.
\end{proposition}
\begin{proof}
The matrix $\Cm_2-\Cm_1\Rm$ has size $r\times (n-k-1)$, hence there
are $\binom{n-k-1}r$ different minors $P_J = \det(\Cm
(\begin{smallmatrix}
    -\Rm\\\Im_{n-k-1}
  \end{smallmatrix})_{\any,J})$. To compute them, we use the
Cauchy-Binet formula for the determinant of a product of non-square matrices:
\begin{equation*}
  \det(\Am\Bm) = \sum_{T\subset\{1..p\}, \#T = r}\det(\Am_{\any,T})\det(\Bm_{T,\any})
\end{equation*}
where $\Am\in\mathbb K^{r\times p}$, $\Bm\in\mathbb K^{p\times r}$,
and $p\ge r$. We apply this formula to $P_J$, and use the
fact that, for $T=T_1\cup T_2$ with $T_1\subset\{1..k+1\}$ and $T_2\subset\{k+2..n\}$,
\begin{eqnarray*}
  \det\left(
  \begin{pmatrix}
    -\Rm\\\Im_{n-k-1}
  \end{pmatrix}_{T_1\cup T_2,J}  \right) &=& 0 \text{ if } T_2 \not\subset J\\
  &=&(-1)^{\Sign{T_2}{J}} \det(\mat R_{T_1, J\backslash T_2}) \text{ if } T_2 \subset J,
\end{eqnarray*}
using the Laplace expansion of this determinant along the last rows, with $\Sign{T_2}{J} = d(k+r)+(d-1)d/2+\sum_{t\in T_2} Pos(t,J)$ where $Pos(t,J)$ is the position of $t$ in $J$, and $d = \#J-\#T_2$.
\qed
\end{proof}
Each polynomial $P_J$ can be expanded into $m$ equations over $\ff q$,
the polynomial $P_J[i]$ being the coefficient of $P_J$ in
$\alpha^{i-1}$. When computing a grevlex Gr\"obner basis of the system
of the $P_J[i]$'s over $\ff q$, with an algorithm like F4 using linear
algebra, the first step consists in computing a basis of the
$P_J[i]$'s over $\Fq$.

It appears that there may be a fall of degree in this first step, in
degree $r$, that produces equations of degree $r-1$. The following
heuristic explains when this fall of degree occurs.
\begin{heuristic}\label{prop:maxminors_structure}
\begin{itemize}
\item{\em Overdetermined case:}
  when $m\binom{n-k-1}{r} \ge \binom{n}r-1$, generically, a 
  degree-$r$ Gr\"obner basis of the projected system
  $\MaxMinors(\cm_2-\cm_1\mat R)=0$ of $m\binom{n-k-1}r$ equations over
  $\ff q$ contains $\binom{n-1}{r-1} - 1$ equations of degree $r-1$, that
  are obtained by linear combinations of the initial equations.
\item{\em Intermediate case:}
  when $\binom{n}r -1 > m\binom{n-k-1}{r} > \binom{n-1}r$, generically a
  degree-$r$ Gr\"obner basis of the projected system
  $MaxMinors(\cm_2-\cm_1\mat R)=0$ contains
  $m\binom{n-k-1}{r} - \binom{n-1}r$ equations of degree $r-1$, that
  are obtained by linear combinations of the initial equations.
\item{\em Underdetermined case:}
  When $m\binom{n-k-1}r \le \binom{n-1}r$, then generically a
  degree-$r$ Gr\"obner basis of the system contains $m\binom{n-k-1}r$ polynomials that are all
  of degree $r$.
  \end{itemize}
\end{heuristic}
\begin{remark}
Here overdetermined/underdetermined refers to the system of maximal minors 
given by the set of equations $\MaxMinors(\cm_2-\cm_1\mat R)=0$
\end{remark}
\begin{remark}
  The degree-$r$ Gr\"obner bases also contain polynomials of degree
  $r$ in the overdetermined and intermediate cases, but we will not
  compute them, as experimentally they bring no speed-up to the
  computation, see Section~\ref{sec:experiments}.
\end{remark}
\begin{proposition}
  Computing the polynomials in a degree-$r$ Gr\"obner basis of the projected
  equations $\MaxMinors$ amounts to solving a linear system with
  $\nu=m\binom{n-k-1}{r}$ equations in $\mu=\binom{n}{r}$ variables, which
  costs \(O(\min(\mu,\nu)^{\omega-2}\mu\nu)\) operations in the base field, where
  $\omega$ is the exponent of matrix multiplication (see
  \cref{sec:results:complexity}).
\end{proposition}
\begin{proof}
  It is possible to view the system $\MaxMinors(\cm_2-\cm_1\mat R)$
  projected over $\ff q$ as a linear system of $\mu=m\binom{n-k-1}r$
  equations, whose variables are the $\nu=\binom{n}r$ unknowns
  $x_T = \det(\cm_{\any,T})$ for all $T\subset \{1..n\}$ of size
  $r$. The matrix associated to this linear system is a matrix $\mat M$ of size
  $\mu \times \nu$ whose coefficient in row
  $(i,J) : i \in \{1..m\}, J \subset \{k+2..n\}, \#J = r$, and column
  $x_T$ is, with $T_2 = T \cap \{k+2..n\}$:
  \begin{eqnarray}
\mat M[(i,J),x_T] &=&
                    \begin{cases}
                      [\alpha^{i-1}](-1)^{\Sign{T_2}{J}} \det(\mat R_{T\cap\{1..k+1\}, J\backslash T_2})&\text{ if }  T_2  \subset J,\\
                      0 & \text{ otherwise.}
\end{cases}\label{eq:MaxMinorsMatrix}
  \end{eqnarray}
  where $[\alpha^{i-1}]\beta$ is the $i$st component of $\beta\in\Fqm$ viewed in the vector space $\Fq^m$ with generator basis $
  \begin{pmatrix}
    1 & \alpha & \dots & \alpha^{m-1}
  \end{pmatrix}$.

  A basis of the vector space generated by the equations 
  $\MaxMinors(\cm_2-\cm_1\mat R)=0$ is given by
  $\tilde{\mat M} \cdot \mat T$ where $\tilde{\mat M}$ is the row
  echelon form of $\mat M$ and $\mat T$ is the column vector formed by
  the polynomials $\det(\cm_{\any,T}) : \#T=r$. As we are searching
  for equations of degree $r-1$, we order the variables $x_T$ such
  that the ones with $1\in T$ that correspond to polynomials
  $\det(\cm_{\any,T})$ of degree $r-1$ are the rightmost entries
  of the matrix.
  \qed
\end{proof}

\cref{prop:maxminors_structure} can be stated in terms of the matrix
$\mat M$.  In the overdetermined case, that is when
$m\binom{n-k-1}{r} \ge \binom{n}r-1$, we expect matrix $\mat M$ to have
rank $\binom{n}r-1$ with high probability.   This rank can not be larger, as the
(left) kernel space of the matrix has dimension 1 (this comes from the
fact that the equations are homogeneous). Hence,
$\tilde{\mat M}\cdot \mat T$ produces $\binom{n-1}r$ equations of
degree $r$, and $\binom{n-1}{r-1}-1$ equations of degree $r-1$, that
have all the shape $\det(\cm_{\any,T})$ or
$\det(\cm_{\any,T}) - \det(\cm_{\any,T_0})$ where $T_0$ corresponds to
the free variable $x_{T_0}$ of the linear system,
$1\in T_0$. In the intermediate and underdetermined cases, we also expect matrix $\mat M$ to be full rank in general, and to be
also full rank on the columns corresponding to the $\cv_T$'s of degree $r$.

\section{Experimental results, complexity bounds, and security}
\label{sec:results}

\subsection{Experimental results}
\label{sec:experiments}

We did various computations for different values of the parameters
$(m,n,k,r)$.  We got our best complexity results by doing the
following steps:

\begin{enumerate}
\item compute the set of equations $\mathcal F$ which comes from  $\begin{pmatrix}
    1 & \alpha & \cdots & \alpha^{m-1}
  \end{pmatrix}
  \sm \left( \mat{C}_2 - \mat{C}_1\mat{R}\right)$ specialised as
  in~\eqref{eqn:oj-modelling-specialized},
\item compute the system $\MaxMinors(\cm_2-\cm_1\mat R)$,
\item compute the matrix $\mat M$ from~\eqref{eq:MaxMinorsMatrix} and
  its echelon form $\tilde{\mat M}$, let $\mathcal J$ be the set of
  the resulting equations of degree $r-1$ in the $\cm$ variables,
\item if $\mathcal J$ is empty, then let $\mathcal J$ be the set of
  equations coming from~$\tilde{\mat M}$ of degree $r$ in the $\cm$
  variables,
\item compute $\mat G$ a reduced degree-$d$ Gr\"obner basis of the
  system $\{\mathcal {F,J,F}_q\}$, where
  \begin{equation*}
 d =
  \begin{cases}
    r & \text{ in the overdetermined case},\\
    r \text{ or } r+1 & \text{ in the intermediate case},\\
    r+2& \text{ in the underdetermined case}.
  \end{cases}
\end{equation*}
\end{enumerate}
The computations are done using \texttt{magma v2.22-2} on a machine with an
Intel\textsuperscript{\textregistered} Xeon\textsuperscript{\textregistered}
2.00GHz processor. Here are the notation used in all
tables:
\begin{itemize}
  \item[$\bullet$] $n_S=(r-1)(m-r)$: the number of variables in $\sm$
  \item[$\bullet$] $n_C=r(n-1)$: the number of variables in $\cm$
  \item[$\bullet$] $n_{eq}=m(n-k-1)$: the number of equations in $\mathcal F$
  \item[$\bullet$] $d:n_{syz}$: the number of equations in $\mathcal J$, where $d$
                    denotes the degree of the equations and $n_{syz}$ the number of
                    them:
  \begin{itemize}
    \item[$\bullet$] $r-1:\binom{n-1}{r-1}-1$ in the overdetermined case
    \item[$\bullet$] $r-1:m\binom{n-k-1}r - \binom{n-1}r$ in the intermediate case
    \item[$\bullet$] $r:m\binom{n-k-1}r$ in the underdetermined case
  \end{itemize}
  \item[$\bullet$]$T_{syz.}$: time of computing the $n_{syz}$ equations of degree
  $r-1$ or $r$ in $\mathcal J$
  \item[$\bullet$]$T_{Gb syz}$: time of the Gr\"obner basis computation of $\{\mathcal J,\mathcal F_q\}$
  \item[$\bullet$]$T_{Gb}$: time of the Gr\"obner basis computation of $\{\mathcal F, \mathcal J, \mathcal F_q\}$
  \item[$\bullet$]$d_{ff}$:  the degree where we observe the first fall of degree
  \item[$\bullet$]$d_{max}$: the maximal degree where some new polynomial is
                    produced by the F4 algorithm
  \item[$\bullet$] ``Max Matrix size'': the size of the largest matrix reduced
                    during the F4 computation, given by \texttt{magma}. We did not take
                    into account the useless steps (the matrices giving no new
                    polynomials) 
\end{itemize}

Table~\ref{tab:LP} page~\pageref{tab:LP} gives our timings on the
parameters proposed in~\cite{LP06}.  For each set of parameters, the
first row of the table gives the timing for the direct computation of
a Gr\"obner basis of $\{ \mathcal F,\mathcal F_q\}$ whereas
the second row gives the timings for the Gr\"obner basis of
$\{ \mathcal {F,J,F}_q\}$. We can see that, apart from very
small parameters, the computation of the equations
$\MaxMinors(\cm_2 - \cm_1 \mat R)$ is negligible compared to the time
of the Gr\"obner basis computation.

Among the proposed parameters, only the $(15,15,8,3)$ was in the case
where the system $\MaxMinors$ is underdetermined. In that case, the
most consuming part of the computation is the Gr\"obner basis of the
system $\MaxMinors$, that depends only on the $\cm$ variables. Once
this computation is done, the remaining Gr\"obner basis of
$\{\mathcal {F,J,F}_q\}$ has a negligible cost.

Table~\ref{tab:n18} page~\pageref{tab:n18} gives timing for different
values of $k$ and $r$, with $m=14$ and $n=18$ fixed. For $r=2$, the
values $k\in\{1..11\}$ correspond to the overdetermined case, the
value $k=12$ to the intermediate one, and $k=13$ to the
underdetermined case. The values $k\in\{1..11\}$ behave all like
$k=11$. As for the parameters from~\cite{LP06}, the hardest cases are
the ones when the system $\MaxMinors$ is underdetermined, where the
maximal degree reached during the computation is $r+2$.  For the
overdetermined cases, the maximal degree is $r$, and for the
intermediate cases, it may be $r$ or $r+1$.

For $r=3$, the overdetermined cases are $k\in\{1..8\}$, $k=9$ is
intermediate and $k\in\{10..11\}$ are underdetermined. Values of
$k\ge 12$ do not allow a unique decoding for $r=3$, the
Gilbert-Varshamov bound being 2 for those values.

For $r=4$ the tradeoffs are $1\le k \le 6$, $k=7$ and $8\le k \le 9$
for the three cases, and for $r=5$, $1\le k \le 5$, $k=6$ and
$7\le k \le 8$. We could not perform the computations for the
intermediate and underdetermined cases, due to a lack of memory.
We also observe that the first fall of degree ($d_{ff}$) does not always
predict the complexity of the computation.

Table~\ref{tab:r3} page~\pageref{tab:r3} gives the timings for a fixed
$r=3$, a ratio $n = 2k$ and various values of $k$. Again, we can
observe that for defavorable cases ($k=6, 7$) the maximal degree is
$r+2$ or $r+1$ rather than $r$, making the computation harder for
small values of $k$ than for larger. 

\begin{table}
  \caption{We compare the behavior of the Gr\"obner basis computation
    for the parameters considered in \cite{LP06a}, with and without
    adding to the system the equations $\mathcal J$.
  }
  \centering \scalebox{0.95}{
    \begin{tabular}{|*{7}{c|}*{6}{r|}|*{6}{c|}|}
      \hline
      $m$ & $n$ & $k$ & $r$ & $n_S$ & $n_C$ & $n_{eq}$& $n_{syz}$ &$T_{syz}$& $T_{Gb syz}$ & $T_{Gb}$ & $d_{ff}$ & $d_{max}$& Max Mat Size  \\
      \hline
      25 & 30 & 15 & 2 & 23 & 58 & 350 &\multicolumn{2}{c|}{} &&0.4 s & 3&3&18550 $\times$19338    \\
\multicolumn{7}{|c|}{}  & 1:28&0.4 s& &0.02 s& 2&2& 1075 $\times$ 749   \\\hline
      30 & 30 & 16 & 2 & 28 & 58 & 390& \multicolumn{2}{c|}{}  &&0.5 s &3 &3& 22620 $\times$ 25288  \\
\multicolumn{7}{|c|}{} &1:28 & 0.4 s& &0.02 s & 2 & 2 &  1260 $\times$ 899   \\\hline
      30 & 50 & 20 & 2 & 28 & 98 & 870 &\multicolumn{2}{c|}{} &&2.2 s & 3 &3 &67860 $\times$ 57898  \\
\multicolumn{7}{|c|}{} &1:48&  3.8 s& & 0.07 s & 2 & 2 & 2324 $\times$ 1499   \\\hline
       50 & 50 & 26 & 2 & 48 & 98 & 1150 & \multicolumn{2}{c|}{}& & 7.4 s & 3 &3 &112700 $\times$ 120148 \\
 \multicolumn{7}{|c|}{} & 1:48&  3.5 s& & 0.2 s & 2 & 2 &  3589$\times$2499   \\\hline
      15 & 15 & 7  & 3 & 24 & 42 & 105 &  \multicolumn{2}{c|}{} & &60.1 s & 4 & 4&77439 $\times$ 153532  \\
 \multicolumn{7}{|c|}{} & 2:90&0.2 s& &  0.06 s & 3 & 3 & 8860 $\times$ 13658  \\\hline
      15 & 15 & 8  & 3 & 24 & 42 & 90 &   \multicolumn{2}{c|}{} & -- & & 4 & $\ge $5 & -- \\
 \multicolumn{7}{|c|}{} &  \underline{3:300} & 0.3 s& 162 s & 0.2 s & 4 & 5 & 191515 $\times $ 457141\\\hline
      20 & 20 & 10 & 3 & 34 & 57 & 180 &   \multicolumn{2}{c|}{} && 450 s & 4 & 4 & 233672 $\times$ 543755  \\
 \multicolumn{7}{|c|}{} &  2:170&1.0 s& & 0.2 s & 3 & 3 & 22124 $\times$ 35087 \\
      \hline
    \end{tabular}
  }\label{tab:LP}
\end{table}

\begin{table}
  \caption{$m=14$ and $n=18$. 
  }
  \centering
  \scalebox{0.95}{
    \begin{tabular}{|*{11}{r|}rcl|r|}
  \hline
  $k$&$r$&$n_{syz}$&$n_S$ & $n_C$ &$n_{eq}$& $T_{Syz.}$ & $T_{Gb syz}$ & $T_{Gb}$ & $d_{ff}$ &$d_{max}$& \multicolumn{3}{c|}{Max Matrix size} &Mem.\\
      \hline
      \hline
11 &2& 1:16 & 12 & 34 & 84 & $<0.1$s && $<0.1$s &2&2& 322&\(\times\) &251& 34 Mo\\
      \hline
12 &2& 1:4 & 12 & 34 & 70 & $<0.1$s && $<0.1$s &3&3& 1820&\(\times\) &2496& 34 Mo\\
      \hline
13 &2& 2:84 & 12 & 34 & 56 & $<0.1$s &32 s & 0 s & 3 & 4 & 231187&\(\times\) &141064& 621 Mo\\
      \hline
      \hline
      8 & 3 & 2:135 & 22 & 51 & 126 & 0.6 s && 0.1 s & 3 &3 & 13179 &\(\times\)& 18604 &34 Mo\\
      \hline
      9 & 3 & 2:104 & 22 & 51 & 112 & 0.5 s && 0.7 s & 3 &3 & 10907~&\(\times\)&~18743   &67 Mo\\
      \hline\hline
      4 & 4 & 3:679 & 30 & 68 & 182 & 12.1 s && 53.7 s &2 &4&314350 &\(\times\)& 650610 &1.3 Go\\
      5 & 4 & 3:679 & 30 & 68 & 168 & 9.4 s && 59.3 s &4 &4&314350 &\(\times\)& 650610  &2.0 Go\\
      6 & 4 & 3:679 & 30 & 68 & 154 & 7.1  s && 69.4 s &4 &4 & 281911 &\(\times\)& 679173 &3.6 Go\\\hline
      \hline
     2 & 5 & 4:2379 & 36 & 85 & 210 & 138.8  s && 27.5 s & 2& 4& 416433 &\(\times\)& 669713& 1.1 Go\\ 
      5 & 5 & 4:2379 & 36 & 85 & 196 &  44.8 s && 5h08 & 2 &  5 &  7642564   &\(\times\)& 30467163 & 253.6 Go\\
\hline
\end{tabular}
}\label{tab:n18}
\end{table}

\begin{table}
  \caption{The parameters are  $r=3$,
    $m=n$, $k=\frac n2$. 
  }
  \centering
    \begin{tabular}{|*{10}{r|}}
  \hline
  $k$&$n_{syz}$ &$n_S$ & $n_C$ &$n_{eq}$& $T_{syz}$& $T_{Gb syz}$ &$T_{Gb}$ & $d_{max}$ & Memory\\\hline
      6 & 3:120 & 18 & 33 & 60 & 0.2s & 117 s & 0.02s & 5 & 4.9 Go\\
      7 & 3:280 & 22 & 39 & 84 & 0.1s & 9.7 s & 0.1s & 4&0.3 Go\\
  \hline
  8 & 2:104 & 26 & 45 & 112 & 0.2s & &0.1s & 3&.04 Go\\
  17 & 2:527 & 62 & 99 & 544 & 34.3s & &4.7s &3& 0.3 Go\\
  27 & 2:1377 & 102 & 159 & 1404 & 650.2s && 161.3s&3& 2.7 Go\\
  37 & 2:2627 & 142 & 219 & 2664 & 5603.6s & &3709.4s &3& 15.0 Go\\
  47 & 2:4277 & 182 & 279 & 4324 & 26503.9s & &26022.6s & 3& 83.0 Go\\
\hline
\end{tabular}\label{tab:r3}
\end{table}

\subsection{Complexity analysis and security over \(\ff{2}\)}
\label{sec:results:complexity}

Now, we give an upper bound on the complexity of our algebraic approach to
solve the $(m,n,k,r)$-decoding problem using the modelling of
\cref{sec:modelling:oj}. The complexity is estimated in terms of the number of
operations in \(\ff{2}\) that the algorithm uses. This allows us to update the
number of bits of security for several cryptosystems, as showed in
\cref{tbl:security}: Loidreau's one \cite{L17}, ROLLO \acite{ROLLO}, and RQC
\acite{RQC2}.
Note that the restriction to \(\ff{2}\) is only there because we want to derive security values. If one works over a larger field \(\Fq\), a similar analysis can be derived. The only change in this case is to consider the relevant number of monomials. Note also that even if \cref{alg:decoding} works over any field, its success probability given in \cref{proposition_SCSS} depends on $q$.
\begin{table}
  \centering
	\begin{tabular}{|c|c|c|c|c|}
	\hline \textbf{Cryptosystem} & \textbf{Parameters $(m,n,k,r)$}
						 					& \ \ \ \ $d=r$ \ \ \ \			&$d=r+1$			& \textbf{Previous} 	 	\\\hline
    Loidreau 		&	$(128,120,80,4)$ 	& \textbf{96.3}					&117.1				&	256						\\\hline\hline\hline
    ROLLO-I-128 	&	$(79,94,47,5)$		& \textbf{114.9} 				&134.5				& 	128						\\\hline
	ROLLO-I-192 	&	$(89,106,53,6)$		& \textbf{142.2}				&162.5				&	192						\\\hline
	ROLLO-I-256 	&	$(113,134,67,7)$	& 174.0							&\textbf{195.3}		&	256 					\\\hline\hline\hline
	ROLLO-II-128	&	$(83,298,149,5)$	& \textbf{132.3} 				&155.4				& 	128						\\\hline
	ROLLO-II-192 	&	$(107,302,151,6)$	& \textbf{161.5}				&185.0				& 	192 					\\\hline
	ROLLO-II-256 	&	$(127,314,157,7)$	& 191.6							&\textbf{215.4}		& 	256 					\\\hline\hline\hline
	ROLLO-III-128	&	$(101,94,47,5)$		& \textbf{117.1} 				&137.2				& 	128 					\\\hline
  	ROLLO-III-192	&	$(107,118,59,6)$	& \textbf{145.7}				&166.6				& 	192 					\\\hline
	ROLLO-III-256	&	$(131,134,67,7)$	& 175.9							&\textbf{197.5}		& 	256 					\\\hline\hline\hline
    RQC-I		  	&	$(97,134,67,5)$		& \textbf{121.1}				&142.0				& 	128 					\\\hline
    RQC-II 			&	$(107,202,101,6)$	& \textbf{154.2}				&176.5				& 	192 					\\\hline
    RQC-III 		&	$(137,262,131,7)$	& 188.4							&\textbf{211.9}		& 	256 					\\\hline
	\end{tabular}
  \caption{Security in bits for several cryptosystems with respect to our
    attack, computed using \cref{eqn:complexity} with $\omega=2.807$, $d=r$ or
    $d=r+1$. The values in bold correspond the most likely maximal degree,
    i.e.~$r$ if \cref{eq:highly-overdetermined} holds and $r+1$ otherwise. The
    last column gives the previous best known security values, based
  on the attack in~\cite{AGHT18}.}
  \label{tbl:security}
\end{table}

Remark that, in \cref{tbl:security}, for the sets of parameters which do not
satisfy \cref{eq:highly-overdetermined}, which correspond to underdetermined
instances, we assume that the system can be solved at $d=r+1$. It is a
conservative choice: in the experiments of Section \ref{sec:experiments}, the
maximal degree is often $r$ for the underdetermined cases.  

The complexity bound follows from the fact that the Gr\"obner basis algorithm
works with Macaulay matrices of degree \(\delta\) for increasing values of
\(\delta\) up to \(d\), the degree for which the Gr\"obner basis is found (see
\cref{sec:syzygies} for a more detailed description). At each of these steps,
the algorithm performs a Gaussian elimination algorithm on a Macaulay matrix
which has at most \(\binom{(m-r)(r-1)+(n-1)r}{\delta}\) columns and fewer rows
than columns. The number of columns is the number of squarefree monomials of
degree \(\delta\) in \((m-r)(r-1)+(n-1)r\) variables.

In general, Gaussian elimination of a $\mu \times \nu$ matrix of rank \(\rho\)
over a field has a complexity of $\OOtsp{\rho^{\omega-2} \mu \nu} \subseteq
\OOtsp{\max(\mu,\nu)^\omega}$ operations in that field \cite{BH74,S00b}. Here,
\(\omega\) is the exponent of matrix multiplication, with naive bounds $2 \le
\omega \le 3$. The best currently known value for $\omega$ is \(\omega \approx
2.37\) \cite{L14b}, by an improvement of Coppersmith-Winograd's algorithm. In
terms of practical performances, the best known method is based on Strassen's
algorithm, which allows one to take $\omega \approx 2.807$, and when the base
field is a finite field, this exponent is indeed observed in practice for
matrices with more than a few hundreds rows and columns.

The Macaulay matrices encountered in the Gr\"obner basis computations we
consider are usually very sparse and exhibit some structure. Some Gaussian
elimination algorithms have been designed specifically for matrices over
\(\ff{2}\) \cite{AB14}, also for sparse matrices \cite{BD16}, and even to take
advantage of the specific structure of Macaulay matrices (see \cite{BEFLM16};
we expect \texttt{Magma}'s closed-source implementation of F${}_4$ to use
similar techniques). However, none of these optimized algorithms has been
proven to reach a complexity which is asymptotically better than the one
mentioned above, apart from speed-ups by constant factors.

As a result, we bound the complexity of the step of degree \(\delta\) in the
Gr\"obner basis computation by that of performing Gaussian elimination on a
$\mu \times \nu$ matrix over \(\ff{2}\), with \(\mu \le \nu =
\binom{(m-r)(r-1)+(n-1)r}{\delta}\); the overall computation then costs
\begin{equation}
  \label{eqn:complexity}
  \OO{\left( \sum_{\delta=0}^{d} \ \binom{(m-r)(r-1)+(n-1)r}{\delta} \right)^\omega \ }
\end{equation} 
operations in \(\ff{2}\). Let us now focus on the case $m=n=2k$ and $r
\approx\sqrt{n}$. Then the complexity of our approach is as in
\cref{eqn:complexity} with \(d=r\). Using a similar analysis, the approach
based on Kipnis-Shamir's modelling has a complexity of
\[
  \OO{\left( \sum_{\delta=0}^{r+2} \binom{km+r(n-r)}{\delta} \right)^\omega \ }
\]
operations. Asymptotically, the dominant term in the former bound is of the
order of $2^{\frac{3}{2} \omega r\log_2(n)}$, to be compared to $2^{2\omega
r\log_2(n)}$ in the Kipnis-Shamir bound. Also, the aforementioned 
combinatorial attacks (\cite{AGHT18}) would have a complexity of the order of 
$2^{\frac{1}{2}rn}$ when $m=n=2k$. 

Finally, note that the complexity bound stated above was derived under
assumptions: in \cref{sec:modelling:oj}, we presented the modelling along with
some assumptions which allowed us to specialize variables a priori and still
ensure that the algorithm of \cref{sec:jac} yields the solution
\(\lambda\err\). In general, the assumption might not hold, that is, the
specific specialization made in \cref{sec:modelling:oj} could be wrong.  We use
\Cref{alg:decoding} in order to specialize more variables: it first uses the
specialization detailed in \cref{sec:modelling:oj}, and if that one fails,
follows on with other similar specializations. This algorithm uses the
subroutine $\mathtt{Solve}(\mat{S},\mat{C},\mat{R})$, which augments the system
as explained in \cref{sec:jac} and returns a solution to \cref{eq:oj-modelling}
if one is found and $\emptyset$ otherwise.

\begin{algorithm}[htbp]
  \DontPrintSemicolon
  \KwData{Matrix $\mat{R}$}
  \KwResult{A solution to the system in \cref{eq:oj-modelling} or $\emptyset$}
  \(\mat{S} = m \times r\) matrix of variables \;
  \(\mat{C} = r \times n\) matrix of variables \;
  Set the first column and the first row of $\mat{S}$ to $[1 \;\; 0 \;\; \cdots \;\; 0]$ \;
  Set a randomly selected column of $\mat{C}$ to $\trsp{[1 \;\; 0 \;\; \cdots \;\; 0]}$ \;
  Choose at random $\lfloor \frac{m-1}{r-1} \rfloor$ disjoint subsets $T_i \subseteq \{2,\hdots,m\}$ of cardinality $r-1$ \;
  \For{$i\gets1$ \KwTo $\lfloor \frac{m-1}{r-1} \rfloor$}{
    Set the $(r-1)\times(r-1)$ submatrix $\mat{S}_{T_i,\{2,\hdots,r\}}$ to $\mat{I}_{r-1}$ \;
    $\mathrm{sol} = \mathtt{Solve}(\mat{S},\mat{C},\mat{R})$       \;
    \lIf{$\mathrm{sol} \neq \emptyset$}{ \Return $\textrm{sol}$ }
  }
  \Return $\emptyset$ \;
  \caption{$(m,n,k,r)$-Decoding}
  \label{alg:decoding}
\end{algorithm}

For positive integers $a$ and $b$ with $a\le b$, we denote by
$p_{q,a,b}:=\prod_{i=0}^{a-1} \left(1-q^{i-b}\right)$ the probability that a
uniformly random matrix in $\matRing{\ff{q}}{a}{b}$ has full rank. 

\begin{proposition}\label{proposition_SCSS}
  Fix integers $m,n,k,r$, and let $c \in \{1,\ldots,\lfloor \frac{m-1}{r-1}
  \rfloor\}$. Suppose that a $(m,n,k,r)$-rank decoding instance is chosen uniformly at
  random, and that the input matrix \(\mat{R}\) is built from this instance. Then,
  the probability that \cref{alg:decoding} makes at most $c$ calls to
  \(\mathtt{Solve}(\mat{S},\mat{C},\mat{R})\) before finding a solution is
  greater than
  \[   \frac{1-q^{-r}}{1-q^{-n}} \left(1-\frac{\left(1-p_{q,r-1,r-1}\right)^c}{p_{q,r-1,m-1}}\right). \]
\end{proposition}
The proof is differed to \hyperref[app:proof_SCSS]{Appendix}.

If one applies this proposition to the cryptosystems mentioned in
\cref{tbl:security}, with at most 5 calls to
\(\mathtt{Solve}(\mat{S},\mat{C},\mat{R})\), \cref{alg:decoding} will return a
solution with a probability always greater than $0.8$; note that for these
instances the quantity $\lfloor \frac{m-1}{r-1} \rfloor$ is greater than 15,
and around 20 for most of them.

In the event where \cref{alg:decoding} returns \(\emptyset\) after $\lfloor
\frac{m-1}{r-1} \rfloor$ calls to \(\mathtt{Solve}(\mat{S},\mat{C},\mat{R})\),
one can run it again until a solution is found. The probabilities mentioned in
the previous paragraph show that for parameters of interest a second run of the
algorithm is very rarely needed.

\section{Conclusion}

In this paper we introduce a new approach for solving the Rank Metric 
Decoding problem with Gr\"obner basis techniques. Our approach is based on adding partial syzygies 
to a newer version of a modelling due to Ourivski and Johansson. 

Overall our analysis shows that our attack, for which we give a general
estimation, clearly outperforms all previous attacks in rank metric for a classical (non quantum)
attacker. In particular we obtain an attack below the claimed security level for all 
rank-based schemes
proposed to the NIST Post-Quantum Cryptography Standardization Process.
Note that there has been some very recent progress \cite{BBCGPSTV20} on the modelling and the attack proposed here. This results in even less complex attacks
and in the removal of the Gr\"obner basis computation step: it is replaced by solving a linear system.
Although our attack and its recent improvement really improve on previous attacks for rank metric,
they meanwhile suffer from two limitations. 

First these attacks do not benefit from
a direct Grover quantum speed-up, unlike combinatorial attacks. 
For the NIST parameters (with the exception of Rollo-I-192 for the latest attack 
\cite{BBCGPSTV20}) the best quantum attacks still remain quantum attacks based on
combinatorial attacks, because of the Grover speed-up. Second, these attacks  need an 
important amount of memory for large parameters.

\section*{Acknowledgements}

This work has been supported by the French ANR projects CBCRYPT
(ANR-17-CE39-0007) and the MOUSTIC project with the support from the European
Regional Development Fund (ERDF) and the Regional Council of Normandie. The
authors would like to thank the anonymous reviewers for their valuable comments
and suggestions, as well as Ray Perlner and Daniel Smith for useful
discussions.

\bibliographystyle{splncs04}

\appendix
\section*{Appendix: Proof of Proposition \ref{proposition_SCSS}}
\label{app:proof_SCSS}

Let $n,m,k,r$ be positive integers such that $n$ and $m$ are both greater than
$r$. Let $E$ be a $\ff{q}$-vector space of $\ff{q^m}$ of dimension $r$
spanned by $\{E_1,E_2,\hdots,E_r\}$ and let $\ev \in \ff{q^m}^{n}$ whose
components generate $E$. By definition, there exists a non-zero coordinate
$e_j$ of $\ev$, and hereafter one focuses on the vector space $\lambda
E=\langle \lambda E_1, \lambda E_2, \hdots, \lambda E_r\rangle$ where
$\lambda=e_j^{-1}$.

Given a basis $(1,\alpha,\dots,\alpha^{m-1})$ of $\ff{q^m}$ over $\ff{q}$, one
can write a basis of $\lambda E$ as a matrix $\mat{S}\in
\matRing{\ff{q}}{m}{r}$.  By construction, $1 \in \lambda E$, so that we can
set the first column and the first row of $\mat{S}$ to the vectors $\trsp{[ 1
\;\; 0 \;\; \cdots \;\; 0] }$ and $ [1 \;\; 0 \;\; \cdots \;\; 0] $. We write
$\mat{\widehat{S}}$ for the remaining $(m-1) \times (r-1)$ block of $\mat{S}$.
One can also express the coordinates of the components of $\lambda \ev$ (with
respect to the basis $\{\lambda E_1,\lambda E_2, \hdots, \lambda E_r\}$) as a
matrix $\mat{C} \in \matRing{\ff{q}}{r}{n}$.  By construction, the $j$-th
column of $\mat{C}$ is the vector $\trsp{[1 \;\; 0 \;\; \cdots \;\; 0]}$.

\cref{lemma_1} estimates the probability to come across an index $j$ such that $e_j$ 
is non-zero. Once such an index is found, \cref{lemma_2} computes the probability 
that \cref{alg:decoding} succeeds in finding a non-singular block in $\mat{\widehat{S}}$.

\begin{lemma}\label{lemma_1}
  With the same notation and hypotheses as above, if an index $j$ is chosen
  uniformly at random in $\{1,\hdots,n\}$, then $e_j$ will be non-zero with
  probability $(1-q^{-r})/(1-q^{-n})$.
\end{lemma} 

\begin{proof} 
	A component $e_j$ of $\ev$ will be non-zero if and only if its corresponding column of coordinates in the matrix $\mat{C}$ is non-zero. 
	If the components of $\ev$ were chosen uniformly at random in the vector space $E$ of dimension $r$, the probability for a random component to be equal to zero would be exactly $q^{-r}$. This is not the case since there is a constraint on $\mat{C}$, more precisely it has to be of rank $r$. 
	
	Taking this into account, we can count the number of full rank matrices in $\matRing{\ff{q}}{r}{n}$ that have a zero column. The ratio between those matrices and all the full rank matrices in $\matRing{\ff{q}}{r}{n}$ is exactly the probability for a column chosen at random in $\mat{C}$ to be zero:
	\[  \prod_{i=0}^{r-1}\frac{q^{n-1}-q^i}{q^n-q^i}=\frac{q^{n-r}-1}{q^n-1}  .\]
	One concludes the proof by taking the complementary event.	\qed
\end{proof}

\begin{lemma}\label{lemma_2}
  Let $c \in \{1,\ldots,\lfloor \frac{m-1}{r-1}\rfloor\}$; with the same notation and hypotheses as above, if $E$ and $e$ are chosen uniformly at random, and 
  if the inverse of a non-zero coordinate of $e$, $\lambda$, is given, then
  at least one of the $c$ disjoint blocks $B_i$ in $\mat{\widehat{S}}$ is not singular with probability greater 
  than 
  \[  1-\frac{(1-p_{q,r-1,r-1})^c}{p_{q,r-1,m-1}}  \]
\end{lemma} 
\begin{proof} Since $\lambda$ is a fixed nonzero element in $\ff{q^m}$ and since $E$ is uniformly random, the vector space $\lambda E$ is also uniformly random. Therefore $\mat{\widehat{S}}$ is a matrix 
  chosen uniformly at random among all the full rank matrices in $\matRing{\ff{q}}{(m-1)}{(r-1)}$.
  The probability that all the $c$ blocks $B_i$ in $\mat{\widehat{S}}$ are singular is then bounded from above by 
  \begin{equation}
    \label{k_non_inv}
    \frac{\left(q^{(r-1)^2}-q^{(r-1)^2}p_{q,r-1,r-1}\right)^c \; q^{(r-1)(m-1-c(r-1))}}{q^{(m-1)(r-1)}p_{q,r-1,m-1}},
  \end{equation}
  which is the ratio between the number of matrices in $\matRing{\ff{q}}{(m-1)}{(r-1)}$ with $c$ singular disjoint blocks and 
  the total amount of full rank matrices in $\matRing{\ff{q}}{(m-1)}{(r-1)}$. It is an upper bound since the number of 
  matrices with $c$ singular blocks includes matrices that are not of full rank. 

  The reader can check that the term (\ref{k_non_inv}) is equal to 
  \[    \frac{\left(1-p_{q,r-1,r-1}\right)^c}{p_{q,r-1,m-1}}.  \]
  The probability that at least one of the $B_i$'s is non-singular is obtained using the complementary probability. \qed
\end{proof}
In \cref{alg:decoding}, the first requirement not to return fail is to 
find an index $j$ such that $e_j$ is non-zero; \cref{lemma_2} gives the probability of this event, that is to 
say $(1-q^{-r})/(1-q^{-n})$.
Once this index is found, the associated vector space $\lambda E$ is distributed 
uniformly among all the vector spaces of $\ff{q^m}$ of dimension $r$ since $E$ is chosen at random. 
Using \cref{lemma_1}, one has a lower bound on the probability that at least one of the 
$c$ block $B_i$'s is non singular. Thus 
the probability of \cref{proposition_SCSS} is 
\[   \frac{1-q^{-r}}{1-q^{-n}} \left(1-\frac{\left(1-p_{q,r-1,r-1}\right)^c}{p_{q,r-1,m-1}}\right). \] \qed

\end{document}